\documentclass[journal]{IEEEtran}

\usepackage{algorithm}
\usepackage{algpseudocode}
\usepackage{graphicx}
\usepackage{epstopdf}
\usepackage{amsmath}
\usepackage{amssymb,latexsym}
\usepackage{accents}
\usepackage{caption}
\usepackage{subcaption}
\usepackage[framemethod=TikZ]{mdframed}
\usepackage{lipsum}
\usepackage[section]{placeins}
\usepackage[noadjust]{cite}
\usepackage{microtype}
\usepackage{lipsum}

\mdfdefinestyle{MyFrame}{%
    linecolor=black,
    outerlinewidth=0.5pt,
    innertopmargin=2pt,
    innerbottommargin=2pt,
    innerrightmargin=4pt,
    innerleftmargin=4pt,
    backgroundcolor=white}

\newtheorem{thm}{Theorem}
\newtheorem{theorem}{Theorem}[section]

\newtheorem{proposition}[theorem]{Proposition}
\newtheorem{corollary}[theorem]{Corollary}

\newcommand*\rc{\color[rgb]{0, 0, 0}}
\newcommand*\rca{\color[rgb]{0, 0, 0}}

\usepackage{etoolbox}
\AtBeginEnvironment{align}{\setcounter{subeqn}{0}}
\newcounter{subeqn} \renewcommand{\thesubeqn}{\theequation\alph{subeqn}}%
\newcommand{\subeqn}{%
  \refstepcounter{subeqn}
  \tag{\thesubeqn}
}

\IEEEoverridecommandlockouts
\begin{document}  
{\title{Joint Power Control in \\ Wiretap Interference Channels}}
\author{
    {\IEEEauthorblockN{Ashkan Kalantari~\IEEEmembership{ Student Member,~IEEE}, Sina Maleki~\IEEEmembership{Member,~IEEE}, 
		\\
		Gan Zheng~\IEEEmembership{Senior Member,~IEEE}, Symeon Chatzinotas~\IEEEmembership{Senior Member,~IEEE}, 
		\\
		and Bj\"{o}rn Ottersten},~\IEEEmembership{Fellow,~IEEE}
    \\
    \thanks{{This work is supported by national Luxembourg projects AFR reference 5798109 and SeMIGod.} {Ashkan Kalantari, Sina Maleki, Symeon Chatzinotas and Bj\"{o}rn Ottersten are with SnT, The University of Luxembourg, Luxembourg (E-mail: {ashkan.kalantari, sina.maleki, symeon.chatzinotas, and bjorn.ottersten}@uni.lu).} 
     {Gan Zheng is with University of Essex, UK, (E-mail: ganzheng@essex.ac.uk). He is also affiliated with SnT.}}
}}
{\maketitle}

\vspace{-1.8cm}
\begin{abstract}
Interference in wireless networks degrades the signal quality at the terminals. However, it can potentially enhance the secrecy rate. This paper investigates the secrecy rate in a two-user interference network where {\rc one of the users, namely user~$1$, requires to establish a confidential connection. User~$1$ wants to prevent an unintended user of the network to decode its transmission}. User~$1$ has to transmit such that its secrecy rate is maximized while the quality of service at the destination of the other user, user~$2$, is satisfied, and both user's power limits are taken into account. We consider two scenarios: 1) user~$2$ changes its power in favor of user~$1$, an altruistic scenario, 2) user~$2$ is selfish and only aims to maintain the minimum quality of service at its destination, an egoistic scenario. It is shown that there is a threshold for user~$2$'s transmission power that only below or above which, depending on the channel qualities, user~$1$ can achieve a positive secrecy rate. Closed-form solutions are obtained in order to perform joint optimal power control. Further, a new metric called secrecy {\rc energy} efficiency is introduced. We show that in general, the secrecy {\rc energy} efficiency of user~$1$ in an interference channel scenario is higher than that of an interference-free channel.
\end{abstract}
\section{Introduction} \label{Sec:Introduction}
Broadcasting information over the same frequency band in wireless networks leads to interference among users. Even in the systems where the spatial dimension is used to concentrate the signal towards the intended destination, the destination may receive interfering signals from other transmitters operating in the same frequency band. Also, due to the expansion and deployment of wireless services, the spectrum is getting scarce~\cite{Olson:2008}. As one possible solution, devices can share the same spectrum which results in interference and degradation of the signal quality. For instance, IEEE standards such as WiFi, Zigbee and Bluetooth share the same frequency band named the industrial, scientific and medical (ISM) band and they may interference with each other~\cite{Chiasserini:2003}. Furthermore, the wireless medium leaves the information vulnerable to unintended users who can potentially decode the message which was meant for other users. Throughout this paper, the words ``wiretapper'', or ``eavesdropper'' refer to the unintended users. While there are higher layer cryptography techniques to secure the data, it is yet possible that a malicious agent breaks into the encryption and gets access to the data~\cite{sklavos:Cryptography:2007}. By intelligently tuning the system parameters using physical layer security techniques, we can prevent the wiretappers from getting access to the information and this way, and further improve the system security along other techniques. Consequently, a specific rate
can be perfectly secured for the users to transmit their data, so that the wiretapper is not able to decode the message. There are efficient coding schemes which can achieve this rate. However, this area is still in its infancy, and the research effort at the moment is inclined in implementing practical codes~\cite{FP7}.

Potentially, the interference can improve the secrecy rate by introducing extra interference at the eavesdropper. The possibility of secure transmission in a multi-user interference channel using interference alignment and secrecy pre-coding is investigated in~\cite{Koyluoglu:2008}. The authors of~\cite{Agrawal:2009} investigate the secrecy rate in a two-user interference channel with an external eavesdropper. They show structured transmission results in a better secrecy rate compared to randomly generated Gaussian codebooks. The authors of~\cite{Koyluoglu:2011} study the secrecy capacity region
for a two-user interference channel in the presence of an external eavesdropper. The users jointly design randomized codebooks and inject noise along with data transmission to improve the secrecy rate. The authors of~\cite{Xiaojun:2011} consider a user who gets helping interference in order to increase its confidentiality against an eavesdropper. The achievable secrecy rate for both discrete memoryless and Gaussian channels is derived. {\rc A two-user interference network with an unintended user is considered in~\cite{Kalantari:2014}. Depending on the channel conditions, bounds on the transmission power of the interfering user is derived such that a positive secrecy rate is sustained for the other user.}

As an example of the interference channel, the effect of interference on the secrecy rate is also investigated
in cognitive radio systems. In cognitive radios, secondary user transmits in the
primary user's operating frequency band when it is not in use. Stochastic geometry is used
in~\cite{Zhihui:2011} to analyze physical layer secrecy in a multiple node cognitive radio network
where an eavesdropper is present. The secrecy outage probability and the secrecy rate of the
primary user is derived while secondary user produces interference. The authors of~\cite{Yiyang:2010}
maximize the secrecy rate for a multiple-antenna secondary user in the presence of an external
eavesdropper while considering the quality of service (QoS) at the primary receiver. Similar
problems to maximize the secrecy rate through beamforming design in cognitive radio are studied
in~\cite{Seongah:2012,Keonkook:2013,Taesoo:2012}.
\subsection{Contributions and main results}
In this work, we investigate the secrecy rate in a {\rc two-user wireless interference network}. Apart from the two users, one of
the idle users {\rc (unintended user)} in this network is a potential eavesdropper. Both nodes transmit in a way so that the secrecy rate is maximized for the first user (user $1$), and the second user (user~$2$) maintains the QoS at its intended destination. {\rc Only user $1$ needs to establish a secure connection and to keep its data secure.} 
For example, in a network with ISM band users, user $1$ and user $2$ can be WiFi and ZigBee transmitters. The ZigBee can be used to send measurement data, which is one of its applications, so its data may not be necessarily important to the {\rc potential} eavesdropper who is interested in WiFi messages.

We study the effect of interference from user $2$ on the secrecy rate of user $1$ in two scenarios, namely altruistic and egoistic scenarios. In the altruistic scenario, we jointly optimize the transmission powers of both users in order to maximize the secrecy rate of user $1$, while maintaining the QoS at user $2$'s destination equal or above a specific threshold. The incentives for {\rc user $2$} to cooperate are twofold: 1) when positive secrecy rate cannot be granted for {\rc user $1$}, it can enjoy an interference-free transmission,
2) {\rc user $1$} adjusts its transmission power to maintain the QoS of {\rc user $2$}'s destination equal or above the threshold. In the egoistic
scenario, the users' powers are still jointly optimized. However, user~$2$ is selfish and only tries to maintain the minimum QoS at the corresponding destination. The
contributions of our work are as follows. {\rca It is shown that by appropriate control of user~$1$'s power, we can make sure that the eavesdropper cannot decode the signal of user~$2$, and thus cannot employ successive interference cancellation (SIC)}. Also, it is shown that the transmitted power from {\rc user~$2$} has a crucial role in achieving a positive secrecy rate for user~$1$. According to the channel conditions, we define the proper power transmission for {\rc user~$2$} to maintain a positive secrecy rate for user~$1$. We develop closed-form expressions to implement joint optimal power control for both users in both altruistic and egoistic scenarios. Finally, a new metric called ``secrecy {\rc energy} efficiency'' is defined, which is the secrecy rate over the consumed power ratio. Using the new metric, it is shown that the interference channel can outperform the single-user channel for specific values of QoS requirements.
\subsection{Related Work}
Inner and outer bounds for the secrecy capacity regions in a two-user interference channel with destinations as
eavesdroppers are investigated in~\cite{Ruoheng:2008}. They showed that the secrecy capacity can be enhanced when one user transmits signal with artificial noise. {\rc Later,~\cite{Ruoheng:2008} was extended to the case when both users transmit
artificial noise along with data in~\cite{Jingge:2010}. As a result, they achieve a larger secrecy rate region when one or both destinations are considered as eavesdropper}. In~\cite{Zang:2008}, an outer bound for secrecy capacity region
is calculated for a two-user one-sided interference channel. {\rc Outer bounds on sum rate of a two-user Gaussian interference channel are studied in~\cite{Xiang:2009} where message confidentiality is important for users. Secrecy capacity region for a two-user MIMO Gaussian interference channel is investigated in~\cite{Fakoorian:2011} where each receiver is a
potential eavesdropper.} A two-user symmetric linear deterministic interference
channel is investigated in~\cite{Mohapatra:2013}. The achievable secrecy rate is investigated when interference cancellation, cooperation, time sharing, and transmission of random bits are used.~{\rc It is shown that sharing random bits achieves a better secrecy rate compared to sharing data bits.} {\rc A two-user MISO interference channel is considered in~\cite{Ni:2014} where beamforming is performed to maintain fair secrecy rate.} {\rc The work in~\cite{Bross:2013} analyzes a two-user interference channel with one-sided noisy feedback. Rate-equivocation region is derived when the second user's message needs to be kept secret.} {\rc The secrecy rate constrained to secrecy rate outage probability and power is maximized by designing robust beamformer in~\cite{Shuai:2014} where a transceiver pair and multiple eavesdroppers constitute a network.}

{\rc A multiple-user interference channel where only one user as a potential eavesdropper receives interference is considered in~\cite{Xiang:2011}. The sum secrecy rate is derived using nested lattice codes.}
{\rc The authors in~\cite{Rabbachin:2012} consider a wireless network comprised of users, eavesdroppers and interfering nodes. It is shown that interference can improve secrecy rate.}
{\rc A communication network comprised of multiple-antenna base stations and single-antenna users is considered in~\cite{Zesong:2014}. The total transmit power is minimized while the signal-to-interference plus noise ratio and equivocation rate for each user is satisfied.}

{\rc In~\cite{Jianwei:2011}, a two-user network with one-sided-interference where each destination is a potential eavesdropper for the other one is studied. Using game
theory, it is concluded that depending on the objective of each pair, the equilibrium can include or exclude the self-jamming strategy.}
{\rc The authors of~\cite{Fakoorian:2013} analyze a two-user MISO Gaussian interference channel where each destination is a potential eavesdropper. Game theory is used to tackle the scenario where each user tries to maximize the difference between its secrecy rate and the secrecy rate of the other user. Beamformers under full and limited channel information are designed at each transmitter to achieve this goal.}

{\rc A transceiver pair is studied in~\cite{Peng:2014} where they try to increase the secrecy rate using an external interferer when a passive eavesdropper is present.} {\rc The authors of~\cite{El-Halabi:2012}
consider a user and an eavesdropper where known interference which only degrades the decoding ability at the eavesdropper is used to enhance the secrecy capacity.} The secrecy capacity and secrecy outage capacity when closest interfering node and multiple interfering {\rc nodes} are separately employed to prevent eavesdropping is studied in~\cite{Rabbachin:2011}. {\rc It is demonstrated that multiple interferes method is superior to the closet interfering method.}
{\rc The exact secure degrees of freedom for different types of Gaussian wiretap channels are discussed in~\cite{Xie:2014} where cooperative jamming from helpers is used.}

{\rc  The equivocation-rate for a cognitive interference network is considered in~\cite{Yingbin:2009} where the primary receiver is a potential eavesdropper and should not decode the secondary message.}
{\rc A MISO transceiver along with multiple single-antenna eavesdroppers are considered in~\cite{LanZhang:2010}. The relationship of the mentioned network with interference cognitive radio network is used to design the transmit covariance matrix.}
{\rc In~\cite{Yongle:2011}, the secondary user causes interferes to both primary destination and eavesdropper. Primary user tries to maintain its secrecy rate while the secondary aims to increase its rate. The achievable pair rate for both users is derived.}

{\rc \subsection{Paper Organization} }
The remainder of the paper is organized as follows. In Section~\ref{sec:System model}, we introduce the
network topology as well as the signal model. The optimization problems for the altruistic and egoistic scenarios are defined and analyzed in
Section~\ref{sec:Problem Formulation cooperative} and Section~\ref{sec:Problem Formulation noncooperative}, respectively. In Section~\ref{sec:num results}, we evaluate the optimal achievable secrecy rate and compare the two-user wiretap interference channel scheme with the single-user wiretap channel as the benchmark. Finally, the conclusions are drawn in Section~\ref{sec:con}.

{\rc \subsection{Notation} }
$A \mathop {\gtrless}\limits_{(2)}^{(1)} 0$ means that $A>0$ when the conditions
of Case $1$ hold and $A<0$ when the conditions of Case $2$ hold. $|\cdot|$ represents the absolute value.
\section{System model} \label{sec:System model}
\subsection{Signal Model}
We consider a wireless interference network consisting of two users denoted by $U_1$ and $U_2$, two destinations denoted by
$D_1$ and $D_2$, and one user as the eavesdropper denoted by $E$. $E$ is assumed to be passive during $U_1$ and $U_2$ transmission and active outside the mentioned period. All nodes employ one antenna for data communication. We denote
by $x_{1}$ and $x_{2}$, the messages which are sent over the same frequency band from $U_{1}$ and $U_{2}$ to $D_{1}$ and $D_{2}$,
respectively. Sharing the same frequency band by the users leads to cross-interference. While the users send data, their signals are wiretapped by the eavesdropper,~$E$. The network setup is depicted in
Fig.~\ref{Fig:System model}. Here, we consider a scenario where $E$ is only interested
in the data sent by one of the users, namely $U_1$. As a result, $x_2$ is considered as an interfering
signal at both $D_{1}$ and $E$.

 {\rc There are two ways in order to carry out the joint power allocation: 1) users send their channel information to a fusion center. At the fusion center, the optimal power values are calculated and sent back to the users separately, 2) one of the users sends its channel information to the other user who calculates the optimal power values and sends the optimal power value to the corresponding user. It can be seen that the first approach consumes more time and number of transmissions compared to the second one. Since $U_1$ is interested in sustaining a positive secrecy rate, it is fair if this user pays the computational cost. Hence, we assume that $U_2$ sends the channels data to $U_1$ and then $U_1$ calculates the optimal power values and sends back the related optimal power value to $U_2$.}
To perform channel estimation in the network, one approach is that the destinations, including the unintended user, send pilots and the transmitters are then able to estimate the required CSIs. After estimating the channels, $U_2$ forwards
the required CSIs to $U_1$. $U_1$ is then responsible to perform the power control and inform $U_2$ of the optimal power that it can transmit. Note that in practice, it is often optimistic to have such a model, as the eavesdroppers are often totally passive. But here, we assume that the eavesdropper is momentarily active, and thus its channel can be estimated and remains unchanged for the optimal power control usage. One practical example of such a scenario is when the eavesdropper is a known user in a network such that $U_1$'s messages should be kept confidential from it.

The received signals at $D_1$ and $D_2$ are as follows
\begin{align}
{y_{{D_1}}} = \sqrt{P_1}{h_{{U_1},{D_1}}}{x_1} + \sqrt{P_2}{h_{{U_2},{D_1}}}{x_2} + {n_{{D_1}}},
\label{eqn:D1 received signal}
\\
{y_{{D_2}}} = \sqrt {{P_2}} {h_{{U_2},{D_2}}}{x_2} + \sqrt {{P_1}} {h_{{U_1},{D_2}}}{x_1} + {n_{{D_2}}},
\label{eqn:D2 received signal}
\end{align}
where $P_1$ and $P_2$ are the power of the transmitted signals by $U_1$ and $U_2$, and ${h_{{U_i},{D_j}}}$ is
the channel gain from each user to the corresponding destination for $i=1, 2$ and $j=1, 2$. The transmission signal from the $i$-th user, and the additive white Gaussian noise at the $i$-th destination are shown by $\sqrt{P_i}x_{i}$ and $n_{D_{i}}$ for $i=1,2$, respectively. The random variables $x_i$ and $n_{D_i}$ are independent and identically distributed (i.i.d.) with~$x_i\sim\mathcal{CN}(0,1)$ and
$n_{D_i}\sim\mathcal{CN}(0,\sigma _n^2)$, respectively, where $\mathcal{CN}$ denotes the complex
normal random variable. {\rc In practice, some signals follow Gaussian distribution such as the amplitude of sample distributions of OFDM signal~\cite{chiueh2008ofdm}. Using a Gaussian distributed signal may not always be optimal, however, our focus is on maximizing the secrecy rate by designing joint optimal power allocation in a specific system model.} The wiretapped signal at $E$ is
\begin{align}
{y_E} = \sqrt{P_1}{h_{{U_1},E}}{x_1} + \sqrt{P_2}{h_{{U_2},E}}{x_2} + {n_E},
\label{eqn:E received signal}
\end{align}
where ${h_{{U_i},{E}}}$ is the channel coefficient from the $i$-th user to the eavesdropper for $i=1, 2$,
and $n_E$ is the additive white Gaussian noise at the eavesdropper with the same distribution
as $n_{D_i}$. The additive white Gaussian noise at different receivers are assumed to be
mutually independent.
\begin{figure}[t]
  \centering
  \includegraphics[width=8.5cm]{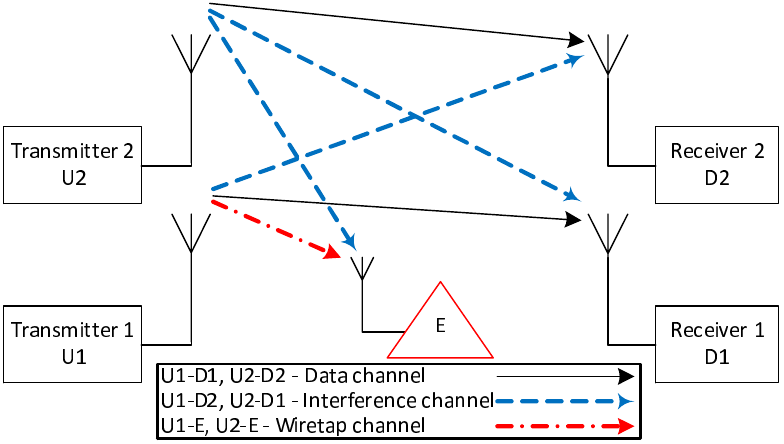}
  \caption{Two-user wireless interference network.}
  \label{Fig:System model}
\end{figure}
{\rca \subsection{Secrecy rate of $U_1$}}
In order to calculate the secrecy rate of $U_1$, we need to first find the rate of $U_1$ without considering the secrecy, and then the rate in which the eavesdropper wiretaps $U_1$. In this paper, we assume that $U_1$ and $U_2$ do not employ SIC. Therefore, using~\eqref{eqn:D1 received signal} and~\eqref{eqn:D2 received signal}, the {\rc rates} for each user
to the corresponding destination can be calculated as
\begin{align}
{I_{{U_1} - {D_1}}} = {\log _2}\left( {1 + \frac{{{P_1}{{\left| {{h_{{U_1},{D_1}}}} \right|}^2}}}{{{P_2}{{\left| {{h_{{U_2},{D_1}}}} \right|}^2} + \sigma _n^2}}} \right),
\label{eqn:S1-D1 capacity}
\\
{I_{{U_2} - {D_2}}} = {\log _2}\left( {1 + \frac{{{P_2}{{\left| {{h_{{U_2},{D_2}}}} \right|}^2}}}{{{P_1}{{\left| {{h_{{U_1},{D_2}}}} \right|}^2} + \sigma _n^2}}} \right).
\label{eqn:S2-D2 capacity}
\end{align}

{\rca The eavesdropper simultaneously receives signals from $U_1$ and $U_2$ which are transmitting in the same frequency band. Hence, the channel from users towards the eavesdropper can be modeled by a multiple-access channel. Assume that the transmission powers of $U_1$ and $U_2$ in a specific time slot are $P_1$ and $P_2$. Then, considering that users employ Gaussian codebooks and the eavesdropper tends to achieve the maximum wiretapping rate from $U_1$, the rate pairs achieved at the eavesdropper are shown in Fig.~\ref{fig:rate_region}~\cite{Tse:1998} which lie on the line from point ``$\mathcal{A}$'' to point ``$\mathcal{D}$''. To wiretap $U_1$ with the maximum achievable rate, the eavesdropper can employ the SIC method~\cite{Tse:WirelessComm}. Using SIC, the eavesdropper first decodes the signal from $U_2$ while considering $U_1$'s signal as interference. Then, considering the fact that the signal from $U_2$ is decoded and known, eavesdropper deducts $U_2$'s signal from the received signal and gets an interference-free signal from $U_1$. In this approach, the rate pairs on the line ``$\mathcal{C}\mathcal{D}$'' are achieved at the eavesdropper if the transmission rate of $U_2$, defined by $R_2$, is lower than the decode-able rate defined at point ``$\mathcal{G}$''. To prevent the eavesdropper from achieving the maximal wiretapping rate, $U_2$'s transmission rate needs to be higher than the decode-able rate at point ``$\mathcal{G}$''. Since users do not coordinate in order to implement time-sharing or rate-splitting, $U_1$'s signal cannot be decoded with the rates which are on the line ``$\mathcal{D}\mathcal{E}$'', and thus it needs to decode $U_1$ considering $U_2$ as the interference with a rate equal to the rate at point ``$\mathcal{E}$''. Therefore, to disable the eavesdropper from performing SIC (i.e., achieving rate at point ``$\mathcal{D}$''), the following condition needs to hold:    
\begin{align}
R_2 =&  {\log _2}\left( {1 + \frac{{{P_2}{{\left| {{h_{{U_2},{D_2}}}} \right|}^2}}}{{{P_1}
{{\left| {{h_{{U_1},{D_2}}}} \right|}^2} + \sigma _n^2}}} \right) 
\nonumber\\
&> {\log _2}\left( {1 + \frac{{{P_2}{{\left| {{h_{{U_2}{,_E}}}} \right|}^2}}}{{{P_1}
{{\left| {{h_{{U_1},E}}} \right|}^2} + \sigma _n^2}}} \right).
\label{eqn:anti-decoding con}
\end{align}  
\begin{figure}[t]
  \centering
  \includegraphics[width=8.5cm]{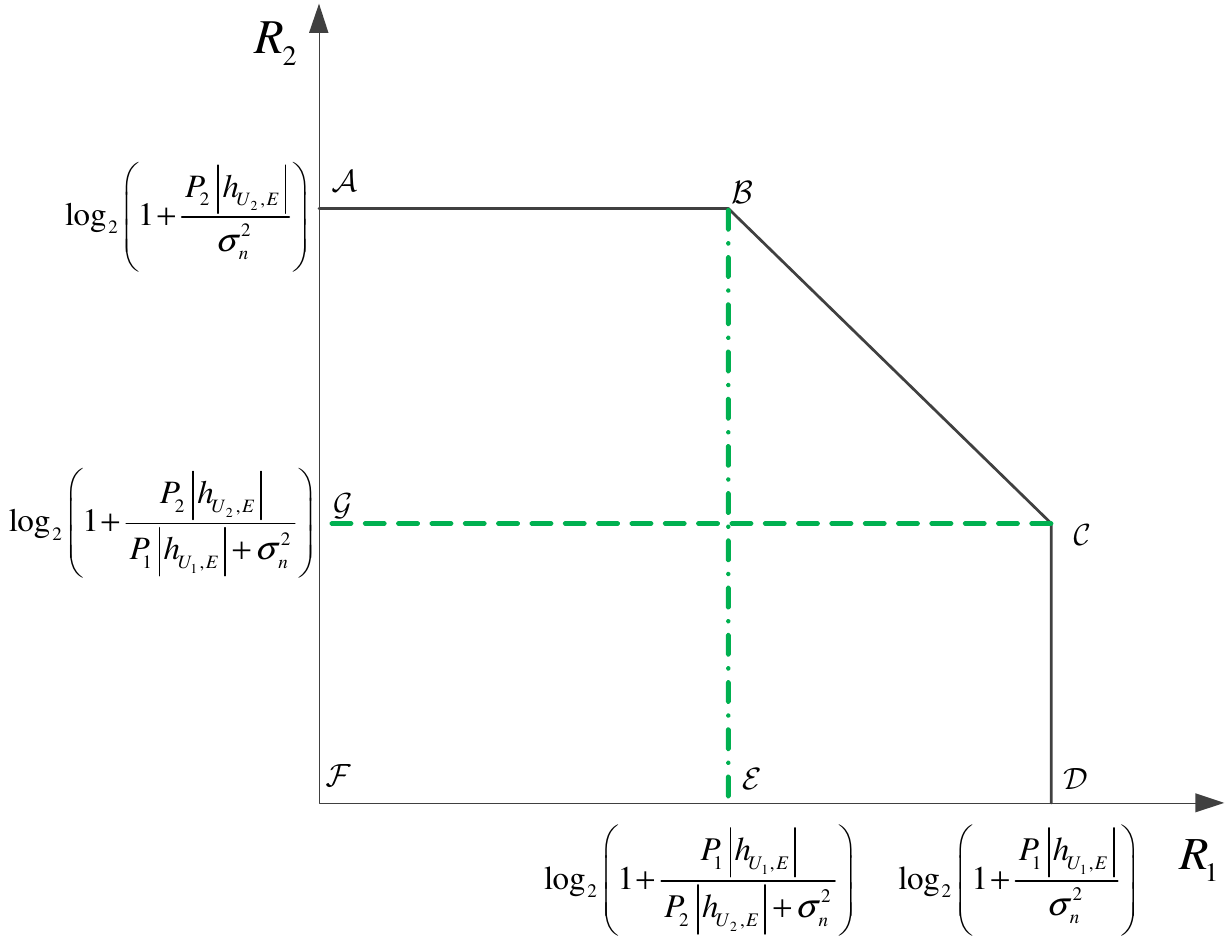}
  \caption{{\rca Maximum achievable rate pairs of a two-user multiple-access fading channel.}}
  \label{fig:rate_region}
\end{figure}
In~\eqref{eqn:anti-decoding con}, the left-hand side is the actual transmission rate of $U_2$ which is equal 
to the decode-able rate at its destination, $D_2$ . If condition~\eqref{eqn:anti-decoding con} is satisfied, the eavesdropper has 
to decode $U_1$'s signal by considering $U_2$'s signal as interference. Interestingly, satisfying condition~\eqref{fig:rate_region} 
just needs $U_1$ to adjust its transmission power and is independent from $P_2$. The condition on $P_1$ to satisfy~\eqref{eqn:anti-decoding con} is derived as:
\begin{alignat}{2}
&{P_1} > \frac{A''}{B''} \qquad \qquad\qquad \,\,\,\, \text{if} \qquad               &&  A'' > 0, B'' > 0, \refstepcounter{equation} \subeqn \label{subeq:AD1}
\\
&{P_1} > 0              \qquad \qquad \,\,\,\, \qquad \,\,\,\,\,\, \text{if} \qquad  &&  A'' < 0, B'' > 0, \subeqn \label{subeq:AD2}
\\
&{P_1}< \frac{A''}{B''}  \qquad \qquad \qquad \,\,\,\,\, \text{if} \qquad               &&  A'' < 0, B'' < 0, \subeqn \label{subeq:AD3}
\\
&{P_1}<0 \,\, ( \text{not feasible})   \qquad \text{if} \qquad            &&  A'' > 0, B'' < 0, \subeqn \label{subeq:AD4}
\end{alignat}
where $A''=\sigma _n^2\left( {{{\left| {{h_{{U_2},E}}} \right|}^2} - {{\left| {{h_{{U_2},{D_2}}}} \right|}^2}} \right)$ and
$B''={{{\left| {{h_{{U_2},{D_2}}}} \right|}^2}{{\left| {{h_{{U_1},E}}} \right|}^2} - {{\left| {{h_{{U_1},{D_2}}}} \right|}^2}
{{\left| {{h_{{U_2},E}}} \right|}^2}}$. As we can see, the channel conditions define weather $U_1$ can block the 
eavesdropper by adjusting its power. For the Cases~\ref{subeq:AD1},~\ref{subeq:AD2}, and~\ref{subeq:AD3}, 
the instantaneous wiretap rate from $U_1$ toward $E$ is obtained by ${I_{{U_1} - E}} = {\log _2}\left( {1 + \frac{{{P_1}{{\left| {{h_{{U_1},E}}} \right|}^2}}}
{{{P_2}{{\left| {{h_{{U_2},E}}} \right|}^2} + \sigma _n^2}}} \right)$, and thus the secrecy rate of $U_1$ in this case is as follows
\begin{align}
{C_{{S_{{U_1}}}}} =& {I_{{U_1} - {D_1}}} - {I_{{U_1} - E}} = {\log _2}\left( {1 + \frac{{{P_1}{{\left| {{h_{{U_1},{D_1}}}} \right|}^2}}}{{{P_2}{{\left| {{h_{{U_2},{D_1}}}} \right|}^2} + \sigma _n^2}}} \right) 
\nonumber\\
&- {\log _2}\left( {1 + \frac{{{P_1}{{\left| {{h_{{U_1},E}}} \right|}^2}}}{{{P_2}{{\left| {{h_{{U_2},E}}} \right|}^2} + \sigma _n^2}}} \right).
\label{eqn:Cs eve block}
\end{align}
For Case~\ref{subeq:AD4}, no power from $U_1$ is capable of preventing the eavesdropper from applying the SIC 
technique and deriving an interference-free version of $U_1$'s signal and thus ${I_{{U_1} - E}} = {\log _2}\left( {1 + \frac{{{P_1}{{\left| {{h_{{U_1},E}}} \right|}^2}}}
{{\sigma _n^2}}} \right)$. This results in the following secrecy rate
\begin{align}
{C_{{S_{{U_2}}}}} =& {I_{{U_1} - {D_1}}} - {I_{{U_1} - E}} = {\log _2}\left( {1 + \frac{{{P_1}{{\left| {{h_{{U_1},{D_1}}}} \right|}^2}}}
{{{P_2}{{\left| {{h_{{U_2},{D_1}}}} \right|}^2} + \sigma _n^2}}} \right) 
\nonumber\\
&- {\log _2}\left( {1 + \frac{{{P_1}{{\left| {{h_{{U_1},E}}} \right|}^2}}}{{\sigma _n^2}}} \right).
\label{eqn:Cs eve dec}
\end{align}
In the next two sections, we formulate and solve the underlying problems so as to find the optimal $P_1$ and $P_2$.}
\vspace{0.5cm}
\section{Problem Formulation: Altruistic Scenario} \label{sec:Problem Formulation cooperative}
In this section, we maximize the secrecy rate of $U_1$ subject to the peak power limits of the users as well as the quality of service (QoS) at $D_2$. {\rca If one of the cases~\ref{subeq:AD1},~\ref{subeq:AD2}, or~\ref{subeq:AD3} holds, using~\eqref{eqn:Cs eve block}, the following secrecy rate optimization is solved:}
\begin{align}
& \mathop {\max }\limits_{{P_1},{P_2}} \,\,\,{C_{{S_{{U_1}}}}}
\nonumber\\
&\,\,  \text{s.\,t.}  \,\,\,\,\,\,\,  {P_1} \le {P_{\max_1 }},
\,{\rca {P_1} \mathop {\gtrless}\limits_{\eqref{subeq:AD3}}^{\eqref{subeq:AD1}}   \omega},
\,{P_2} \le {P_{\max_2 }},
\, {I_{{U_2} - {D_2}}} \geq \beta,
\label{eqn:SC Opt 1 Ins}
\end{align}
where $\beta$ is the minimum {\rca required} data rate for $U_{2}$ {\rca and $\omega =\frac{A''}{B''}$}. {\rca In Case~\ref{subeq:AD2}, any $P_1$ ensures that the eavesdropper cannot employ SIC. Therefore, no additional constraint over $P_1$ is necessary.} {\rca For Case~\ref{subeq:AD4}, using~\eqref{eqn:Cs eve dec}, the following secrecy rate optimization problem should be solved
\begin{align}
& \mathop {\max }\limits_{{P_1},{P_2}} \,\,\,{C_{{S_{{U_2}}}}}
\nonumber\\
& \,\,  \text{s.\,t.}  \,\,\,\,\,\,\,  {P_1} \le {P_{\max_1 }},
\, {P_2} \le {P_{\max_2 }},
\, {I_{{U_2} - {D_2}}} \geq \beta.
\label{eqn:SC Opt 1 Ins eve dec}
\end{align}
We first solve~\eqref{eqn:SC Opt 1 Ins} and then~\eqref{eqn:SC Opt 1 Ins eve dec}}. By inserting~\eqref{eqn:Cs eve block} 
into~\eqref{eqn:SC Opt 1 Ins}, we obtain
\begin{align}
&\mathop {\max }\limits_{{P_1},{P_2}} \,\,\,{\log _2}\left( {\frac{{1 + \frac{{{P_1}{{\left| {{h_{{U_1},{D_1}}}} \right|}^2}}}
{{{P_2}{{\left| {{h_{{U_2},{D_1}}}} \right|}^2} + \sigma _n^2}}}}{{1 + \frac{{{P_1}{{\left| {{h_{{U_1},E}}} \right|}^2}}}
{{{P_2}{{\left| {{h_{{U_2},E}}} \right|}^2} + \sigma _n^2}}}}} \right)
\nonumber\\
& \,\,  \text{s.\,t.}  \,\,\,\,\,\,\,  {P_1} \le {P_{\max_1 }},
\,\,{\rca {P_1} \mathop {\gtrless}\limits_{\eqref{subeq:AD3}}^{\eqref{subeq:AD1}}   \omega },
\,\,{P_2} \le {P_{\max_2 }},
\nonumber\\
& \,\,\,\,\,\,\,\,\,\,\,\,\,\,\,\,\, \frac{{{P_2}{{\left| {{h_{{U_2},{D_2}}}} \right|}^2}}}
{{{P_1}{{\left| {{h_{{U_1},{D_2}}}} \right|}^2} + \sigma _n^2}}  \geq  \gamma,
\label{eqn:SC Opt 2 Ins}
\end{align}
where $\gamma$ is ${2^{\beta}-1}$. Since $\log$ is a monotonic increasing function of its
argument, we can just maximize the argument and thus we rewrite~\eqref{eqn:SC Opt 2 Ins} as
\begin{align}
&\mathop {\max }\limits_{{P_1},{P_2}} \,\,\,\,{\frac{{1 + \frac{{{P_1}{{\left| {{h_{{U_1},{D_1}}}} \right|}^2}}}
{{{P_2}{{\left| {{h_{{U_2},{D_1}}}} \right|}^2} + \sigma _n^2}}}}{{1 + \frac{{{P_1}{{\left| {{h_{{U_1},E}}} \right|}^2}}}
{{{P_2}{{\left| {{h_{{U_2},E}}} \right|}^2} + \sigma _n^2}}}}}
\nonumber\\
& \,\,  \text{s.\,t.}\,\,\,\,\,\,\,{P_1} \le {P_{\max_1 }},
\,\, {\rca {P_1} \mathop {\gtrless}\limits_{\eqref{subeq:AD3}}^{\eqref{subeq:AD1}}   \omega   },
\,\, {P_2} \le {P_{\max_2 }},
\nonumber\\
& \,\,\,\,\,\,\,\,\,\,\,\,\,\,\,\,\,  \frac{{{P_2}{{\left| {{h_{{U_2},{D_2}}}} \right|}^2}}}
{{{P_1}{{\left| {{h_{{U_1},{D_2}}}} \right|}^2} + \sigma _n^2}}  \geq \gamma.
\label{eqn:SC Opt 3 Ins A}
\end{align}
Considering that the objective function is neither convex, nor concave, solving problem~\eqref{eqn:SC Opt 3 Ins A} is difficult. As a
result, we shall adopt a two-step approach in order to solve~\eqref{eqn:SC Opt 3 Ins A}. First, we consider $P_2$ to be
fixed and derive the optimal value for $P_1$, and then we replace the obtained $P_1$ in~\eqref{eqn:SC Opt 3 Ins A}
and solve the optimization problem for $P_2$.
\subsection{Optimizing $P_1$ for a Given $P_2$} \label{subsub P1 P2 Ins}
For this case,~\eqref{eqn:SC Opt 3 Ins A} is reduced to
\begin{align}
&\mathop {\max }\limits_{{P_1}} \,\,\,\,{\frac{{1 + \frac{{{P_1}{{\left| {{h_{{U_1},{D_1}}}} \right|}^2}}}
{{{P_2}{{\left| {{h_{{U_2},{D_1}}}} \right|}^2} + \sigma _n^2}}}}{{1 + \frac{{{P_1}{{\left| {{h_{{U_1},E}}} \right|}^2}}}
{{{P_2}{{\left| {{h_{{U_2},E}}} \right|}^2} + \sigma _n^2}}}}}
\nonumber\\
& \,\,  \text{s.\,t.} \,\,\,\,\,\,\,  {P_1} \le {P_{\max_1 }}, \, 
{\rca {P_1} \mathop {\gtrless}\limits_{\eqref{subeq:AD3}}^{\eqref{subeq:AD1}}   \omega }, \, 
{P_1} \le \frac{{{P_2}{{\left| {{h_{{U_2},{D_2}}}} \right|}^2} - \gamma \sigma _n^2}}{{\gamma {{\left| {{h_{{U_1},{D_2}}}} \right|}^2}}}.
\label{eqn:SC Opt 3 Ins}
\end{align}
In order to solve~\eqref{eqn:SC Opt 3 Ins}, first, we find the range of $P_2$ for which the objective function in~\eqref{eqn:SC Opt 3 Ins} is always
positive, i.e., a positive secrecy rate can be achieved. In the following theorem,
we outline the related bounds on $P_{2}$ where the positive secrecy rate is obtained.
\begin{thm}
{\rc Assume an interference network similar to the one mentioned in Fig.~\ref{Fig:System model} along with the assumptions on power limits and the QoS}. In order to achieve a positive secrecy rate, $P_2$ should satisfy the following bounds:
\begin{alignat}{2}
&{P_2} > \frac{A}{B} \qquad \text{if} \qquad  &&  A > 0, B > 0, \refstepcounter{equation} \subeqn \label{subeq:P1}
\\
&{P_2} > 0           \qquad \,\,\, \text{if} \qquad  &&  A < 0, B > 0, \subeqn \label{subeq:P3}
\\
&{P_2}< \frac{A}{B}  \qquad \text{if} \qquad  &&  A < 0, B < 0, \subeqn \label{subeq:P4}
\end{alignat}
where $A={ \sigma_n^2 \left( {{{\left| {{h_{{U_1},E}}} \right|}^2} - {{\left| {{h_{{U_1},{D_1}}}} \right|}^2}} \right)}$ and
$B={\left| {{h_{{U_1},{D_1}}}} \right|^2}{\left| {{h_{{U_2},E}}} \right|^2} -
{\left| {{h_{{U_2},{D_1}}}} \right|^2}{\left| {{h_{{U_1},E}}} \right|^2}$. Further, for $A > 0, B < 0$,
irrespective of the value of $P_{2}$, no positive secrecy rate can be obtained for $U_{1}$.
\label{thm:Pos sec Ins}
\end{thm}
\begin{proof}
The proof is given in Appendix~\ref{Appen Propo Ins}.
\end{proof}

One immediate conclusion of Theorem~\ref{thm:Pos sec Ins} is given by the following corollary
which can be considered as the most important result of this paper.
\begin{corollary}
In a wiretap interference channel as in Fig.~\ref{Fig:System model}, where the goal is to obtain a positive
secrecy rate for $U_1$, the possibility of achieving a positive secrecy rate is independent from the value of $P_{1}$, and depends on the value of $P_{2}$ and the conditions of the channels.
\label{cor: thm pos sec}
\end{corollary}

Now that we have defined the required conditions for $P_2$ to achieve a positive secrecy rate, we investigate the
optimal value of $P_1$, denoted by $P^\star_1$ for a given $P_{2}$. If we take the derivative of the objective
function in~\eqref{eqn:SC Opt 3 Ins} with respect to $P_1$, we see that the conditions on $P_2$ to have a
monotonically increasing, referred to as Case 1, or
decreasing, referred to as Case 2, are the same as the conditions to have a positive or negative
secrecy rate, respectively. These conditions are summarized as follows
\begin{alignat}{2}
&{P_{2_{(1)}}} > \frac{A}{B}, \, {P_{2_{(2)}}} < \frac{A}{B}  \qquad \text{if} \qquad && A > 0, B > 0,
\nonumber\\
&{P_{2_{(1)}}} = \emptyset, \, {P_{2_{(2)}}} >0               \qquad \,\,\,\,\,\, \text{if}  \qquad && A > 0, B < 0,
\nonumber\\
&{P_{2_{(1)}}} > 0, \, {P_{2_{(2)}}} = \emptyset              \qquad \,\,\,\,\,\, \text{if}  \qquad && A < 0, B > 0,
\nonumber\\
&{P_{2_{(1)}}}< \frac{A}{B}, \, {P_{2_{(2)}}} > \frac{A}{B}   \qquad \text{if}  \qquad && A < 0, B < 0,
\label{derivative P1}
\end{alignat}
where $P_{2_{(1)}}$ refers to the required power in Case 1, $P_{2_{(2)}}$ refers
to the required power in Case 2 and $\emptyset$ denotes the empty set. According
to Theorem~\ref{thm:Pos sec Ins}, and the conditions in~\eqref{derivative P1}, the global optimal
values for $P_1$ in Cases 1 and 2 are defined as
\begin{enumerate}
  \item If the objective function in~\eqref{eqn:SC Opt 1 Ins} is monotonically increasing, then
  \begin{align}
  P^\star_1=\min\left\{ {\rca \chi},\frac{{{P_2}{{\left| {{h_{{U_2},{D_2}}}} \right|}^2} - \gamma \sigma _n^2}}{{\gamma {{\left| {{h_{{U_1},{D_2}}}} \right|}^2}}}\right\}.
  \label{case 1 Ins}
  \end{align}
{\rca where $\chi=P_{\max_1}$ for Cases~\ref{subeq:AD1} and~\ref{subeq:AD2}, 
$\chi  = \min \left\{ {{P_{{{\max }_1}}},\omega } \right\}$ for Case~\ref{subeq:AD3}.}
  \item If the objective function in~\eqref{eqn:SC Opt 1 Ins} is monotonically decreasing,
  then $P^\star_1=0$. This could also be concluded from the fact that when a positive secrecy rate cannot be granted,
  $U_1$ should be turned off. \label{case 2 Ins}
\end{enumerate}
\subsection{Optimizing $P_2$ for a Given $P_1$}
We insert the $P^{\star}_{1}$ obtained in Subsection~\ref{subsub P1 P2 Ins} into~\eqref{eqn:SC Opt 3 Ins}, and try to obtain the optimal value for $P_{2}$. First, we decompose the optimal
answer of $P_1$ in~\eqref{case 1 Ins} into two different answers as follows
\begin{align}
P_1^ \star  = \left\{ {\begin{array}{*{20}{c}}
{\rca \chi} & {P_2} \geq \frac{{\gamma \left( {{\rca \chi} {{\left| {{h_{{U_1},{D_2}}}} \right|}^2} + \sigma _n^2} \right)}}{{{{\left| {{h_{{U_2},{D_2}}}} \right|}^2}}},\\
{\frac{{{P_2}{{\left| {{h_{{U_2},{D_2}}}} \right|}^2} - \gamma \sigma _n^2}}{{\gamma {{\left| {{h_{{U_1},{D_2}}}} \right|}^2}}}}&{P_2} < \frac{{\gamma \left( {{\rca \chi}
{{\left| {{h_{{U_1},{D_2}}}} \right|}^2} + \sigma _n^2} \right)}}{{{{\left| {{h_{{U_2},{D_2}}}} \right|}^2}}}.
\end{array}} \right.
\label{P Opt case 1 & 2}
\end{align}
Using Theorem~\ref{thm:Pos sec Ins} and according to the two resulting cases in~\eqref{P Opt case 1 & 2},
we can break~\eqref{eqn:SC Opt 3 Ins A} into two problems in order to optimize $P_2$, respectively, as follows
\begin{align}
&\mathop {\max }\limits_{{P_2}} \,\,\, \frac{{1 + \frac{{{P_{{{\max }_1}}}{{\left| {{h_{{U_1},{D_1}}}} \right|}^2}}}{{{P_2}
{{\left| {{h_{{U_2},{D_1}}}} \right|}^2} + \sigma _n^2}}}}{{1 + \frac{{{P_{{{\max }_1}}}{{\left| {{h_{{U_1},E}}} \right|}^2}}}{{{P_2}{{\left| {{h_{{U_2},E}}} \right|}^2} + \sigma _n^2}}}}
\nonumber\\
&\,\,\text{s.\,t.}\,\,\,\,\,\,{P_2} \le {P_{\max_2 }},
\,\, {P_2} \geq \frac{{\gamma \left( {{\rca \chi} {{\left| {{h_{{U_1},{D_2}}}} \right|}^2} + \sigma _n^2} \right)}}{{{{\left| {{h_{{U_2},{D_2}}}} \right|}^2}}} = \lambda_1,
\nonumber\\
& \,\,\,\,\,\,\,\,\,\,\,\,\,\,\,\, {P_2} \mathop {\gtrless}\limits_{\eqref{subeq:P4}}^{\eqref{subeq:P1}} \frac{{\sigma _n^2 \left( {{{\left| {{h_{{U_1},E}}} \right|}^2} - {{\left| {{h_{{U_1},{D_1}}}} \right|}^2}} \right)}}
{{{{\left| {{h_{{U_1},{D_1}}}} \right|}^2}{{\left| {{h_{{U_2},E}}} \right|}^2} - {{\left| {{h_{{U_2},{D_1}}}} \right|}^2}{{\left| {{h_{{U_1},E}}} \right|}^2}}} = \varphi_1,
\label{eqn:SC Opt 4 Ins}
\end{align}
and
\begin{align}
&\mathop {\max }\limits_{{P_2}} \,\,\, \frac{{1 + \frac{{\left( {{P_2}{{\left| {{h_{{U_2},{D_2}}}} \right|}^2} - \gamma \sigma _n^2} \right){{\left| {{h_{{U_1},{D_1}}}} \right|}^2}}}
{{\gamma {{\left| {{h_{{U_1},{D_2}}}} \right|}^2}\left( {{P_2}{{\left| {{h_{{U_2},{D_1}}}} \right|}^2} + \sigma _n^2} \right)}}}}
{{1 + \frac{{\left( {{P_2}{{\left| {{h_{{U_2},{D_2}}}} \right|}^2} - \gamma \sigma _n^2} \right){{\left| {{h_{{U_1},E}}} \right|}^2}}}
{{\gamma {{\left| {{h_{{U_1},{D_2}}}} \right|}^2}\left( {{P_2}{{\left| {{h_{{U_2},E}}} \right|}^2} + \sigma _n^2} \right)}}}}
\nonumber\\
&\,\,\text{s.\,t.}\,\,\,\,\,\,{P_2} \le {P_{\max_2 }},
\,\,{P_2} < \frac{{\gamma \left( {{{\rca \chi}} {{\left| {{h_{{U_1},{D_2}}}} \right|}^2} + \sigma _n^2} \right)}}{{{{\left| {{h_{{U_2},{D_2}}}} \right|}^2}}} = \lambda_1,
\nonumber\\
&\,\,\,\,\,\,\,\,\,\,\,\,\,\,\,\,{P_2} \ge \frac{{\gamma \sigma _n^2}}{{{{\left| {{h_{{U_2},{D_2}}}} \right|}^2}}}= \lambda_2,
\nonumber\\
&\,\,\,\,\,\,\,\,\,\,\,\,\,\,\,\,\,{P_2} \mathop {\gtrless}\limits_{\eqref{subeq:P4}}^{\eqref{subeq:P1}} \frac{{\sigma _n^2\left( {{{\left| {{h_{{U_1},E}}} \right|}^2} - {{\left| {{h_{{U_1},{D_1}}}} \right|}^2}} \right)}}
{{{{\left| {{h_{{U_1},{D_1}}}} \right|}^2}{{\left| {{h_{{U_2},E}}} \right|}^2} - {{\left| {{h_{{U_2},{D_1}}}} \right|}^2}{{\left| {{h_{{U_1},E}}} \right|}^2}}} = \varphi_1,
\label{eqn:SC Opt 5 Ins}
\end{align}
for $A \mathop {\gtrless}\limits_{\eqref{subeq:P4}}^{\eqref{subeq:P1}} 0$ and $B \mathop {\gtrless}\limits_{\eqref{subeq:P4}}^{\eqref{subeq:P1}} 0$. For the case $A<0$ and $B>0$ which is represented by~\eqref{subeq:P3},
the last constraint in~\eqref{eqn:SC Opt 4 Ins} and~\eqref{eqn:SC Opt 5 Ins} {\rc is} removed from the problem since with any positive
value for $P_2$, $U_1$ can have a positive secrecy rate. Also for $A>0$ and $B<0$, the secrecy rate is simply zero since $P_1=0$. {\rc Furthermore, the numerator and denumerator in~\eqref{eqn:SC Opt 5 Ins} have the possibility to become less than unit and this leads to a negative rate}. The constraint in~\eqref{eqn:SC Opt 5 Ins} which is placed one to the last, ensures that the data and wiretap {\rc rates} do not go below zero.

We discuss the feasibility conditions of~\eqref{eqn:SC Opt 4 Ins} and~\eqref{eqn:SC Opt 5 Ins} to
derive the feasibility domain, $p_2$, in Proposition~\ref{Pro:feasibility ins}.
\begin{proposition}
The feasibility domain for the problems~\eqref{eqn:SC Opt 4 Ins} and~\eqref{eqn:SC Opt 5 Ins} denoted by $p_{2}$ is defined as follows
\begin{enumerate}
  \item Problem~\eqref{eqn:SC Opt 4 Ins}: For case~\eqref{subeq:P1}, we should have $\max\left\{ { {\lambda _1},\sup {\varphi _1}} \right\} \le {P_{{{\max }_2}}}$ which leads to ${p_2} = [\max \left\{ { {\lambda _1},\sup {\varphi _1}} \right\},{P_{{{\max }_2}}}]$. For
  case~\eqref{subeq:P4}, we should have $\min \left\{ {\inf {\varphi _1},{P_{{{\max }_2}}}} \right\} \ge {\lambda _1}$ which leads to ${p_2} = [ {\lambda _1}, \min \left\{ {\inf {\varphi _1},{P_{{{\max }_2}}}} \right\} ]$.

  \item Problem~\eqref{eqn:SC Opt 5 Ins}: For case~\eqref{subeq:P1}, we should have $\max \left\{ {\sup {\varphi _1},{\lambda _2}} \right\} \le \min \left\{ {\inf {\lambda _1},{P_{{{\max }_2}}}} \right\}$ which leads to $\left[ {\max \left\{ {\sup {\varphi _1},{\lambda _2}} \right\},\min \left\{ {\inf {\lambda _1},{P_{{{\max }_2}}}} \right\}} \right]$. For case~\eqref{subeq:P4}, we should have $\min \left\{ {\inf {\varphi _1},\inf {\lambda _1},{P_{{{\max }_2}}}} \right\} \ge {\lambda _2}$ which leads to ${p_2} = [\min \left\{ {\inf {\varphi _1},\inf {\lambda _1},{P_{{{\max }_2}}}} \right\},{\lambda _2}]$.
\end{enumerate}
\label{Pro:feasibility ins}
\end{proposition}
\begin{proof}
The proof is straightforward, thus was omitted.
\end{proof}
If both~\eqref{eqn:SC Opt 4 Ins} and~\eqref{eqn:SC Opt 5 Ins} are feasible at the same time, we select the $P_2^\star$ and the corresponding secrecy
rate from the problem which results in a higher secrecy rate. Here, we provide a generic closed-form solution
depending on the channels' conditions in Theorems~\ref{thm 1 ins} and~\ref{thm 2 ins}
for~\eqref{eqn:SC Opt 4 Ins} and~\eqref{eqn:SC Opt 5 Ins}, respectively.
\begin{thm}
Assume $a = P_{\max_1 } {\left| {{h_{{U_1},{D_1}}}} \right|^2}$, $b = {\left| {{h_{{U_2},{D_1}}}} \right|^2}$,
$c = P_{\max_1 } {\left| {{h_{{U_1},E}}} \right|^2}$, $d = {\left| {{h_{{U_2},E}}} \right|^2}$, $C=b-d$,
$D=b\left( {c + \sigma _n^2} \right) - d\left( {a + \sigma _n^2} \right)$, $E=-B{P_{{max}_1}}={bc - ad}$,
$F= cd\sigma _n^2 - a\left( {  b\left( {c + \sigma _n^2} \right) - cd} \right)$, $G=\frac{AP_{{max}_1}}{\sigma _n^2}=c-a$, $\alpha={\min\left( {\inf {\varphi _1},{P_{{{\max }_2}}}} \right)}$, and $\beta=\max \left\{ { {\lambda _1},\sup {\varphi _1}} \right\}$. Also, suppose
that~\eqref{eqn:SC Opt 4 Ins} is feasible. Then,~\eqref{eqn:SC Opt 4 Ins} is solved as follows:
\begin{enumerate}
  \item If $CD<0$
\begin{enumerate}
  \item If $A<0$ and $E>0$
    \begin{align}
&P_2^ \star  = \alpha
\end{align}
  \item If $E<0$
  \begin{align}
P_2^ \star  = \begin{cases}
\beta          & A>0
\\
{\lambda _1} & A<0
\end{cases}
\end{align}
\end{enumerate}
  \item If $ CD > 0 $
  \begin{enumerate}
    \item If $A<0$, $E>0$ and $F<0$
    \begin{align}
    P_2^ \star = \mathop {\arg }\limits_{{P_2}} \mathop {\max }\limits_{{P_2} \in \left\{ {{\lambda _1},\alpha } \right\}} {C_s}
  \end{align}
    \item If $E<0$ and $F>0$
    \begin{align}
   \hspace{-1.5cm}  P_2^ \star  =
    \begin{cases}
P_{2C} & P_{2C} \in p_2
\\
\mathop {\arg }\limits_{{P_2}} \mathop {\max }\limits_{{P_2} \in \left\{ {\beta ,{P_{{{\max }_2}}}} \right\}} {C_s} & A>0, P_{2C} \notin p_2
\\
\mathop {\arg }\limits_{{P_2}} \mathop {\max }\limits_{{P_2} \in \left\{ {{\lambda _1},{P_{{{\max }_2}}}} \right\}} {C_s} & A<0, P_{2C} \notin p_2
  \end{cases}
  \end{align}

      \item If $E>0$, $F>0$ and $G<0$
    \begin{align}
     \hspace{-1.5cm}  P_2^ \star  =
    \begin{cases}
\mathop {\arg }\limits_{{P_2}} \mathop {\max }\limits_{{P_2} \in \left\{ {{P_{2C}},{\lambda _1},\alpha } \right\}} {C_s}                  & P_{2C}  \in p_2
\\
\mathop {\arg }\limits_{{P_2}} \mathop {\max }\limits_{{P_2} \in \left\{ {{\lambda _1},\alpha } \right\}} {C_s}                               & P_{2C} \notin p_2
  \end{cases}
  \end{align}

      \item If $E<0$, $F<0$ and $G>0$
    \begin{align}
    \hspace{-1.5cm} P_2^ \star  =
    \begin{cases}
\mathop {\arg }\limits_{{P_2}} \mathop {\max }\limits_{{P_2} \in \left\{ {{P_{2C}},\beta ,{P_{{{\max }_2}}}} \right\}} {C_s}        &       P_{2C} \in p_2
\\
\mathop {\arg }\limits_{{P_2}} \mathop {\max }\limits_{{P_2} \in \left\{ {\beta ,{P_{{{\max }_2}}}} \right\}} {C_s}                & P_{2C} \notin p_2
  \end{cases}
  \end{align}

      \item If $E<0$, $F<0$ and $G<0$
    \begin{align}
    P_2^ \star  = \lambda _1
  \end{align}
  \end{enumerate}
\end{enumerate}
where $C_s$ is the objective function in~\eqref{eqn:SC Opt 4 Ins},
$P_{2C}={\frac{{ - 2bdG\sigma _n^2 - \sqrt \Delta  }}{{2bdE}}}$,
and $\Delta  = 4abcdCD\sigma _n^2$.
\label{thm 1 ins}
\end{thm}
\vspace{0.1cm}
\begin{IEEEproof}
The proof is given in Appendix~\ref{Appen thm 1 ins}.
\end{IEEEproof}
\vspace{0.1cm}
\begin{thm}
Assume $e = {\left| {{h_{{U_1},{D_1}}}} \right|^2}$, $f = {\left| {{h_{{U_2},{D_2}}}} \right|^2}$, $g = {\left| {{h_{{U_1},D_2}}} \right|^2}$,
$h = {\left| {{h_{{U_2},D_1}}} \right|^2}$, $i = {{{\left| {{h_{{U_1},{E}}}} \right|}^2}}$, $j={{{\left| {{h_{{U_2},{E}}}} \right|}^2}}$, $H=h-j$, $\delta=\min \left( {\inf {\lambda _1},{P_{{{\max }_2}}}} \right)$, $\kappa=\min \left\{ {\inf {\lambda _1},\inf {\varphi _1}} , {{P_{{{\max }_2}}}} \right\}$, $\mu={\max \left\{ {\sup {\varphi _1},{\lambda _2}} \right\}}$, $I= { - fi + gh\gamma  - \left( {hi + gj} \right)\gamma  + e\left( {f + j\gamma } \right)}$, 
\\
$J= { - g{h^2}i\gamma \left( {f + j\gamma } \right) + e\left( {{f^2}i\left( { - h + j} \right) + fg{j^2}\gamma  + gh{j^2}{\gamma ^2}} \right)}$,
\\
$K = { - gi\left( {f + j\gamma } \right) + e\left( {fg + gh\gamma  - hi\gamma  + ij\gamma } \right)}$,
\\
$L={ - ghi\left( {f + j\gamma } \right) + e\left( {fhi + fgj - fij + ghj\gamma } \right)}$.
\\
Also, suppose that~\eqref{eqn:SC Opt 5 Ins} is feasible. Then,~\eqref{eqn:SC Opt 5 Ins} can be solved as follows
\begin{enumerate}
 \item If $HI < 0$
\begin{enumerate}
  \item If $J>0$
   \begin{align}
P_2^ \star  = \begin{cases}
\delta & A > 0, B>0
\\
\delta & A < 0, B>0
\\
\kappa  &  A < 0, B<0
\end{cases}
  \end{align}
  \item If $J<0$
   \begin{align}
P_2^ \star  = \begin{cases}
\mu & A > 0, B>0
\\
\lambda_2 & A < 0, B>0
\\
\lambda_2 &   A < 0, B<0
\end{cases}
  \end{align}
\end{enumerate}
  \item If $HI > 0$
\begin{enumerate}
  \item If $J>0$ and $K<0$
   \begin{align}
P_2^ \star  = \begin{cases}
\mathop {\arg }\limits_{{P_2}} \mathop {\max }\limits_{{P_2} \in \left\{ {\mu ,\delta } \right\}} {C_s}               &   A > 0, B>0
\\
\mathop {\arg }\limits_{{P_2}} \mathop {\max }\limits_{{P_2} \in \left\{ {{\lambda _2},\delta } \right\}} {C_s}     &  A < 0, B>0
\\
\mathop {\arg }\limits_{{P_2}} \mathop {\max }\limits_{{P_2} \in \left\{ {{\lambda _2},\kappa } \right\}} {C_s} & A < 0, B<0
\end{cases}
  \end{align}
  \item  If $J<0$ and $K>0$
   \begin{align}
\hspace{-1cm} P_2^ \star  = \begin{cases}
P_{2C}                                 &  P_{2C} \in p_2
\\
\mathop {\arg }\limits_{{P_2}} \mathop {\max }\limits_{{P_2} \in \left\{ {\mu ,\delta } \right\}} {C_s}                       & A > 0, B>0, P_{2C} \notin p_2
\\
\mathop {\arg }\limits_{{P_2}} \mathop {\max }\limits_{{P_2} \in \left\{ {{\lambda _2},\delta } \right\}} {C_s}             & A < 0, B>0, P_{2C} \notin p_2
\\
\mathop {\arg }\limits_{{P_2}} \mathop {\max }\limits_{{P_2} \in \left\{ {{\lambda _2},\kappa } \right\}} {C_s} & A < 0, B<0, P_{2C} \notin p_2
\end{cases}
  \end{align}
  \item  If $J>0$, $K>0$ and $L<0$ or $J<0$, $K<0$ and $L>0$
   \begin{align}
\hspace{-1cm} P_2^ \star  = \begin{cases}
\mathop {\arg }\limits_{{P_2}} \mathop {\max }\limits_{{P_2} \in \left\{ {{P_{2C}},\mu ,\delta } \right\}} {C_s}              &   A > 0, B>0, P_{2C} \in p_2
\\
\mathop {\arg }\limits_{{P_2}} \mathop {\max }\limits_{{P_2} \in \left\{ {{P_{2C}},{\lambda _2},\delta } \right\}} {C_s} &   A < 0, B>0, P_{2C} \in p_2
\\
\mathop {\arg }\limits_{{P_2}} \mathop {\max }\limits_{{P_2} \in \left\{ {{P_{2C}},{\lambda _2},\kappa } \right\}} {C_s} & A < 0, B<0, P_{2C} \in p_2
\\
\mathop {\arg }\limits_{{P_2}} \mathop {\max }\limits_{{P_2} \in \left\{ {\mu ,\delta } \right\}} {C_s} &   A > 0, B>0, P_{2C} \notin p_2
\\
\mathop {\arg }\limits_{{P_2}} \mathop {\max }\limits_{{P_2} \in \left\{ {{\lambda _2},\delta } \right\}} {C_s} & A < 0, B>0, P_{2C} \notin p_2
\\
\mathop {\arg }\limits_{{P_2}} \mathop {\max }\limits_{{P_2} \in \left\{ {{\lambda _2},\kappa } \right\}} {C_s} & A < 0, B<0, P_{2C} \notin p_2
\end{cases}
  \end{align}
     \item  If $J>0$, $K>0$ and $L>0$
   \begin{align}
P_2^ \star  = \begin{cases}
\delta &  A > 0, B>0
\\
\delta &  A < 0, B>0
\\
\kappa  & A < 0, B<0
\end{cases}
  \end{align}
  \item  If $J<0$, $K<0$ and $L<0$
   \begin{align}
P_2^ \star  = \begin{cases}
\mu &  A > 0, B>0
\\
\lambda _2 & A < 0, B>0
\\
\lambda _2 & A < 0, B<0
\end{cases}
  \end{align}

\end{enumerate}
\end{enumerate}
where $C_s$ is the objective function in~\eqref{eqn:SC Opt 5 Ins} and $P_{2C}=\frac{{ - 2 \sigma _n^2 \gamma L - \sqrt \Delta  }}{{2J}}$,
and $\Delta  = 4egiHI{{\left( {\sigma _n^2} \right)^4}}\gamma \left( {f + h\gamma } \right)\left( {f + j\gamma } \right)$.
\label{thm 2 ins}
\end{thm}
\begin{IEEEproof}
The proof can be obtained in the similar way to that of Theorem~\ref{thm 1 ins}.
\end{IEEEproof}
{\rca For problem~\eqref{eqn:SC Opt 1 Ins eve dec}, the optimal solution of $P_1$ is as~\eqref{case 1 Ins} when 
$\chi=P_{\max_1}$. The closed-form solution for the $P_2$ is given in the following theorem.}
{\rca \begin{thm}
 Assume $a = {\left| {{h_{{U_1},{D_1}}}} \right|^2}$, $b = {\left| {{h_{{U_2},{D_2}}}} \right|^2}$,
$c = {\left| {{h_{{U_1},D_2}}} \right|^2}$, $d = {\left| {{h_{{U_2},D_1}}} \right|^2}$, 
$e = {\left| {{h_{{U_1},E}}} \right|^2}$, $A=b(a - e)\sigma  + d( - e\gamma \sigma  + c\gamma \sigma )$, $B=2bde(a\gamma \sigma  - c\gamma \sigma )$, $C= - bce\gamma {\sigma ^2} + a(bc\gamma {\sigma ^2} + d\gamma \sigma ( - e\gamma \sigma  + c\gamma \sigma ))$, $\psi  = \frac{{\sigma _n^2\left( {{{\left| {{h_{{U_1},{D_1}}}} \right|}^2} - {{\left| {{h_{{U_1},E}}} \right|}^2}} \right)}}{{{{\left| {{h_{{U_1},E}}} \right|}^2}{{\left| {{h_{{U_2},{D_1}}}} \right|}^2}}}$, ${\lambda _1} = \frac{{\gamma \left( {{P_{{{\max }_1}}}{{\left| {{h_{{U_1},{D_2}}}} \right|}^2} + \sigma _n^2} \right)}}{{{{\left| {{h_{{U_2},{D_2}}}} \right|}^2}}}$, ${\lambda _2} = \frac{{\gamma \sigma _n^2}}{{{{\left| {{h_{{U_2},{D_2}}}} \right|}^2}}}$, and $\varsigma  = \min \left( {{P_{{{\max }_2}}},\inf \psi ,\inf {\lambda _1}} \right)$. Then, optimal $P_2$ is given as follows:
\begin{enumerate}
  \item If $A<0$
    \begin{align}
&P_2^ \star  = \lambda_2
\end{align}
  \item If $ A > 0 $
  \begin{enumerate}
    \item If $C>0$ or $B>0$ and $C<0$
    \begin{align}
   \hspace{-1.5cm}  P_2^ \star  =
    \begin{cases}
P_{2C} & P_{2C} \in p_2
\\
\mathop {\arg }\limits_{{P_2}} \mathop {\max }\limits_{{P_2} \in \left\{ {\lambda_2 ,\psi } \right\}} {C_s} & P_{2C} \notin p_2
  \end{cases}
  \end{align}
    \item If $B<0$ and $C<0$
    \begin{align}
    P_2^ \star = \mathop {\arg }\limits_{{P_2}} \mathop {\max }\limits_{{P_2} \in \left\{ {{\lambda _2},\psi } \right\}} {C_s}
  \end{align}
\end{enumerate}
\end{enumerate}
where $C_s$ is the secrecy rate, ${P_{2C}} = \frac{{ - B - \sqrt \Delta  }}{{2D}}$, $\Delta  = 4Aabcde\gamma (d\gamma \sigma  + b\sigma )$, $D=- bde(ab + cd\gamma )$, 
and $p_2$ is the feasibility domain of the problem.
\label{thm alr eve dec}
\end{thm}}
\vspace{-0.1cm}
{\rca \begin{proof}
The proof is similar to that of Theorem~\ref{thm 1 ins}, thus was omitted.
\end{proof}}
\section{Problem Formulation: Egoistic Scenario} \label{sec:Problem Formulation noncooperative}
In this section, we develop closed-form solutions for the case when $U_2$ is selfish from the
view point of $U_1$'s secrecy rate, and adjusts its transmission power just to meet its QoS, i.e.,
SINR=$\gamma$. Later, we compare this case with respect to the altruistic scenario. {\rca If one of the Cases~\ref{subeq:AD1},~\ref{subeq:AD2}, or
~\ref{subeq:AD3} holds and $U_2$ is selfish,~\eqref{eqn:SC Opt 3 Ins} can be written as}
\begin{align}
&\mathop {\max }\limits_{{P_1},{P_2}} \,\,\,{\frac{{1 + \frac{{{P_1}{{\left| {{h_{{U_1},{D_1}}}} \right|}^2}}}
{{{P_2}{{\left| {{h_{{U_2},{D_1}}}} \right|}^2} + \sigma _n^2}}}}{{1 + \frac{{{P_1}{{\left| {{h_{{U_1},E}}} \right|}^2}}}
{{{P_2}{{\left| {{h_{{U_2},E}}} \right|}^2} + \sigma _n^2}}}}}
\nonumber\\
&\,\, \text{s.\,t.} \,\, {P_1} \le {P_{\max_1 }},
\, {\rca {P_1} \mathop {\gtrless}\limits_{\eqref{subeq:AD3}}^{\eqref{subeq:AD1}}   \omega},
\,  {P_2} \le {P_{\max_2 }},
\, \frac{{{P_2}{{\left| {{h_{{U_2},{D_2}}}} \right|}^2}}}
{{{P_1}{{\left| {{h_{{U_1},{D_2}}}} \right|}^2} + \sigma _n^2}}  = \gamma.
\label{eqn:SC Opt Ins limited power}
\end{align}
{\rca In Case~\ref{subeq:AD2}, any $P_1$ ensures that the eavesdropper cannot employ SIC, 
so no additional constraint over $P_1$ is necessary. For Case~\ref{subeq:AD4}, the problem is solved as follows
\begin{align}
&\mathop {\max }\limits_{{P_1},{P_2}} \,\,\,\frac{{1 + \frac{{{P_1}{{\left| {{h_{{U_1},{D_1}}}} \right|}^2}}}{{{P_2}{{\left| {{h_{{U_2},{D_1}}}} \right|}^2} + \sigma _n^2}}}}{{1 + \frac{{{P_1}{{\left| {{h_{{U_1},E}}} \right|}^2}}}{{\sigma _n^2}}}}
\nonumber\\
&\,\, \text{s.\,t.} \,\, {P_1} \le {P_{\max_1 }},
\,  {P_2} \le {P_{\max_2 }},
\, \frac{{{P_2}{{\left| {{h_{{U_2},{D_2}}}} \right|}^2}}}
{{{P_1}{{\left| {{h_{{U_1},{D_2}}}} \right|}^2} + \sigma _n^2}}  = \gamma.
\label{eqn:SC Opt Ins limited power eve dec}
\end{align}}
{\rca We first solve~\eqref{eqn:SC Opt Ins limited power} and then~\eqref{eqn:SC Opt Ins limited power eve dec}.} 
Using the last constraint in~\eqref{eqn:SC Opt Ins limited power}, we can directly derive the solution for $P_2$ as ${P_2} = \gamma
\frac{{\left( {{P_1}{{\left| {{h_{{U_1},{D_2}}}} \right|}^2} + \sigma _n^2} \right)}}{{{{\left| {{h_{{U_2},{D_2}}}} \right|}^2}}}$ and replace it with the
corresponding value. Consequently, we can rewrite~\eqref{eqn:SC Opt Ins limited power} as
\begin{align}
&\mathop {\max }\limits_{{P_1}} \,\,\, \frac{{1 + \frac{{{P_1}{{\left| {{h_{{U_1},{D_1}}}} \right|}^2}}}{{\gamma \frac{{\left( {{P_1}{{\left| {{h_{{U_1},{D_2}}}} \right|}^2} + \sigma _n^2} \right)}}{{{{\left| {{h_{{U_2},{D_2}}}} \right|}^2}}}{{\left| {{h_{{U_2},{D_1}}}} \right|}^2} + \sigma _n^2}}}}{{1 + \frac{{{P_1}{{\left| {{h_{{U_1},E}}} \right|}^2}}}{{\gamma \frac{{\left( {{P_1}{{\left| {{h_{{U_1},{D_2}}}} \right|}^2} + \sigma _n^2} \right)}}{{{{\left| {{h_{{U_2},{D_2}}}} \right|}^2}}}{{\left| {{h_{{U_2},E}}} \right|}^2} + \sigma _n^2}}}}
\nonumber\\
&\,\, \text{s.\,t.}  \,\, {P_1} \le {P_{{{\max }_1}}},
\,\, {\rca {P_1} \mathop {\gtrless}\limits_{\eqref{subeq:AD3}}^{\eqref{subeq:AD1}}   \omega},
\,\, {P_1} \le \frac{{{P_{{{\max }_2}}}{{\left| {{h_{{U_2},{D_2}}}} \right|}^2} - \gamma \sigma _n^2}}{{\gamma {{\left| {{h_{{U_1},{D_2}}}} \right|}^2}}}.
\label{eqn:SC Opt Ins limited power P1}
\end{align}
Since the minimum value for the secrecy rate is zero, the objective function in~\eqref{eqn:SC Opt Ins limited power P1} should be greater or equal to one.
Proposition~\ref{Propo Ins lim pow} gives the required condition on $P_1$ in order to have a positive secrecy
rate. According to the channel conditions, these constraints should be added
to~\eqref{eqn:SC Opt Ins limited power P1}.
\begin{proposition}
In order for the objective function in~\eqref{eqn:SC Opt Ins limited power P1} to result in a non-negative secrecy rate, $P_1$ should have the following bounds:
\begin{alignat}{2}
&{P_1} > \frac{A'}{B'} \qquad \text{if} \qquad && A' > 0, B' > 0, \refstepcounter{equation} \subeqn \label{subeq:P1 lim}
\\
&{P_1} > 0             \qquad \,\,\,\,\, \text{if} \qquad && A' < 0, B' > 0, \subeqn \label{subeq:P3 lim}
\\
&{P_1} < \frac{A'}{B'}  \qquad \text{if} \qquad && A' < 0, B' < 0, \subeqn \label{subeq:P4 lim}
\end{alignat}
where $A'={\left( {\left( {1 + c} \right)d - b\left( {1 + e} \right)} \right)\sigma _n^2}$,
$B'={a\left( {be - cd} \right)}$. Also, for the case $A' > 0$ and $B' < 0$, irrespective of the value for $P_2$,
no positive secrecy rate is possible for $U_1$. The values for $b$, $c$, $d$ and $e$ are defined
in Theorem~\ref{thm 1 Lim Pow}.
\label{Propo Ins lim pow}
\end{proposition}
\vspace{0.1cm}
\begin{proof}
The proof is similar to that of Theorem~\ref{thm:Pos sec Ins}, thus was omitted.
\end{proof}

According to~Proposition~\ref{Propo Ins lim pow}, we can rewrite~\eqref{eqn:SC Opt Ins limited power P1} as
\begin{align}
&\mathop {\max }\limits_{{P_1}} \,\,\, \frac{{1 + \frac{{{P_1}{{\left| {{h_{{U_1},{D_1}}}} \right|}^2}}}{{\gamma \frac{{\left( {{P_1}{{\left| {{h_{{U_1},{D_2}}}} \right|}^2} + \sigma _n^2} \right)}}{{{{\left| {{h_{{U_2},{D_2}}}} \right|}^2}}}{{\left| {{h_{{U_2},{D_1}}}} \right|}^2} + {\sigma _n^2}}}}}{{1 + \frac{{{P_1}{{\left| {{h_{{U_1},E}}} \right|}^2}}}{{\gamma \frac{{\left( {{P_1}{{\left| {{h_{{U_1},{D_2}}}} \right|}^2} + \sigma _n^2} \right)}}{{{{\left| {{h_{{U_2},{D_2}}}} \right|}^2}}}{{\left| {{h_{{U_2},E}}} \right|}^2} + {\sigma_n^2}}}}}
\nonumber\\
& \,\, \text{s.\,t.}  \,\,\,\,\,\,\,\,\, {P_1} \le {P_{{{\max }_1}}},
\,\, {\rca {P_1} \mathop {\gtrless}\limits_{\eqref{subeq:AD3}}^{\eqref{subeq:AD1}}   \omega},
\,\, {P_1}  \mathop {\gtreqless}\limits_{\eqref{subeq:P4 lim}}^{\eqref{subeq:P1 lim}} \frac{A'}{B'}=\varphi_3,
\nonumber\\
&  \,\,\,\,\,\,\,\,\,\,\,\,\,\,\,\,\,\,\,  {P_1} \le \frac{{{P_{{{\max }_2}}}{{\left| {{h_{{U_2},{D_2}}}} \right|}^2} - \gamma \sigma _n^2}}{{\gamma {{\left| {{h_{{U_1},{D_2}}}} \right|}^2}}}=\lambda_3.
\label{eqn:SC Opt Ins limited power P1 1}
\end{align}
Assuming that~\eqref{eqn:SC Opt Ins limited power P1 1} is feasible, we give a closed-form solution
for~\eqref{eqn:SC Opt Ins limited power P1 1} in Theorem~\ref{thm 1 Lim Pow}.
\begin{thm}
Assuming $a ={\left| {{h_{{U_1},{D_2}}}} \right|^2}$, $b = {\left| {{h_{{U_1},{D_1}}}} \right|^2}$,
$c = \gamma \frac{{{{\left| {{h_{{U_2},{D_1}}}} \right|}^2}}}{{{{\left| {{h_{{U_2},{D_2}}}} \right|}^2}}}$, $d = {\left| {{h_{{U_1},{E}}}} \right|^2}$,
$e = \gamma \frac{{{{\left| {{h_{{U_2},E}}} \right|}^2}}}{{{{\left| {{h_{{U_2},{D_2}}}} \right|}^2}}}$, $Q=c-e$,
$R=-(1 + c) d + a (c - e) + b (1 + e)$, $S=- a{c^2}d\left( {1 + e} \right) + b\left( {e\left( {d + ae} \right) + c\left( { - d + a{e^2}} \right)} \right)$,
$T=\frac{-A'}{\sigma _n^2}=-(1 + c) d + b (1 + e)$, $U=\frac{B'}{a}=be - cd$, 
$\eta  = \min \left\{ {{P_{{{\max }_1}}},{\lambda _3}} \right\}$, {\rca $\eta ' = \min \left\{ {\eta ,\omega } \right\}$}, $\theta=\min \left\{ {P_{{{\max }_1}}},  \lambda_3 ,  \varphi_3 \right\}$, and {\rca $\theta ' = \min \left\{ {\theta ,\omega {\rm{ }}} \right\}$}~\eqref{eqn:SC Opt Ins limited power P1} can be solved as follows
\begin{enumerate}
  \item If $QR < 0$

  \begin{enumerate}

    \item If $S>0$, {\rca $A''>0,B''>0$}, (or {\rca $A''<0,B''>0$})
    \begin{align}
P_1^\star = \begin{cases}
\eta & Q>0, R<0
\\
\theta & B'<0, Q<0, R>0
\\
\eta & B'>0, Q<0, R>0
\end{cases}
    \end{align}
		
		    \item If $S>0$, {\rca $A''<0,B''<0$}
    \begin{align}
P_1^\star = \begin{cases}
\eta' & Q>0, R<0
\\
\theta' & B'<0, Q<0, R>0
\\
\eta' & B'>0, Q<0, R>0
\end{cases}
    \end{align}

    \item If $S<0$, {\rca $A''>0,B''>0$}
    \begin{align}
P_1^\star = \begin{cases}
{\rca \max \left\{ {{\varphi _3},\omega } \right\}}   & B'>0, Q>0, R<0
\\
{\rca \omega} & Q<0, R>0
\end{cases}
    \end{align}

		    \item If $S<0$, {\rca $A''<0,B''<0$} (or {\rca $A''<0,B''>0$})
    \begin{align}
P_1^\star = \begin{cases}
\varphi _3& B'>0, Q>0, R<0
\\
0 & Q<0, R>0
\end{cases}
    \end{align}

  \end{enumerate}

  \item If $QR > 0$

  \begin{enumerate}

	\item If $S>0$, $T<0$ and $B'>0$
		{\rca
		    \begin{align}
P_1^ \star = \begin{cases}
\mathop {\arg }\limits_{{P_1}} \mathop {\max }\limits_{{P_1} \in \left\{ { \max \left\{ {{\varphi _3},\omega } \right\} ,\eta } \right\}} {C_s}  & A''>0, B'' > 0
\\
\mathop {\arg }\limits_{{P_1}} \mathop {\max }\limits_{{P_1} \in \left\{ {{\varphi _3},\eta' } \right\}} {C_s}     & A''<0, B'' < 0
\\
\mathop {\arg }\limits_{{P_1}} \mathop {\max }\limits_{{P_1} \in \left\{ {{\varphi _3},\eta } \right\}} {C_s}     & A''<0, B''>0
    \end{cases}
    \end{align}
}

       \item If $S<0$, $T>0$, {\rca $A''>0$, $B''>0$}
    \begin{align}
P_1^ \star = \begin{cases}
P_{1C}                                     & P_{1C} \in p_1
\\
\mathop {\arg }\limits_{{P_1}} \mathop {\max }\limits_{{P_1} \in \left\{ {{\rca \omega},\eta } \right\}} {C_s}  & P_{1C} \notin p_1, B'>0
\\
\mathop {\arg }\limits_{{P_1}} \mathop {\max }\limits_{{P_1} \in \left\{ {{\rca \omega},\theta } \right\}} {C_s}  & P_{1C} \notin p_1, B'<0
\end{cases}
    \end{align}
		
	{\rca	
		       \item If $S<0$ and $T>0$, {\rca $A''<0$, $B''<0$}
    \begin{align}
P_1^ \star = \begin{cases}
P_{1C}                                     & P_{1C} \in p_1
\\
\mathop {\arg }\limits_{{P_1}} \mathop {\max }\limits_{{P_1} \in \left\{ {0, {\rca \eta' } } \right\}} {C_s}  & P_{1C} \notin p_1, B'>0
\\
\mathop {\arg }\limits_{{P_1}} \mathop {\max }\limits_{{P_1} \in \left\{ {0,{\rca \theta' } } \right\}} {C_s}  & P_{1C} \notin p_1, B'<0
\end{cases}
    \end{align}
		}

			{\rca	
		       \item If $S<0$ and $T>0$, {\rca $A''<0$, $B''>0$}
    \begin{align}
P_1^ \star = \begin{cases}
P_{1C}                                     & P_{1C} \in p_1
\\
\mathop {\arg }\limits_{{P_1}} \mathop {\max }\limits_{{P_1} \in \left\{ {0, {\rca \eta } } \right\}} {C_s}  & P_{1C} \notin p_1, B'>0
\\
\mathop {\arg }\limits_{{P_1}} \mathop {\max }\limits_{{P_1} \in \left\{ {0,{\rca \theta } } \right\}} {C_s}  & P_{1C} \notin p_1, B'<0
\end{cases}
    \end{align}
		}

    \item If $S>0$, $T>0$, $U<0$ and {\rca $A''>0$, $B''>0$}
    \begin{align}
P_1^ \star = \begin{cases}
\mathop {\arg }\limits_{{P_1}} \mathop {\max }\limits_{{P_1} \in \left\{ {{P_{1C}},{\rca \omega},\theta } \right\}} {C_s}     & P_{1C} \in p_1
\\
\mathop {\arg }\limits_{{P_1}} \mathop {\max }\limits_{{P_1} \in \left\{ {{\rca \omega},\theta } \right\}} {C_s} & P_{1C} \notin p_1
    \end{cases}
    \end{align}

		{\rca
		    \item If $S>0$, $T>0$, $U<0$ and {\rca $A''<0$, $B''<0$}
    \begin{align}
P_1^ \star = \begin{cases}
\mathop {\arg }\limits_{{P_1}} \mathop {\max }\limits_{{P_1} \in \left\{ {{P_{1C}},0,\theta' } \right\}} {C_s}     & P_{1C} \in p_1
\\
\mathop {\arg }\limits_{{P_1}} \mathop {\max }\limits_{{P_1} \in \left\{ {0,\theta' } \right\}} {C_s} & P_{1C} \notin p_1
    \end{cases}
    \end{align}

}

		{\rca
		    \item If $S>0$, $T>0$, $U<0$ and {\rca $A''<0$, $B''>0$}
    \begin{align}
P_1^ \star = \begin{cases}
\mathop {\arg }\limits_{{P_1}} \mathop {\max }\limits_{{P_1} \in \left\{ {{P_{1C}},0,\theta } \right\}} {C_s}     & P_{1C} \in p_1
\\
\mathop {\arg }\limits_{{P_1}} \mathop {\max }\limits_{{P_1} \in \left\{ {0,\theta } \right\}} {C_s} & P_{1C} \notin p_1
    \end{cases}
    \end{align}

}

        \item If $S<0$, $T<0$, $U>0$, {\rca $A''>0$, $B''>0$} (or {\rca $A''<0$, $B''>0$})
    \begin{align}
P_1^ \star = \begin{cases}
\mathop {\arg }\limits_{{P_1}} \mathop {\max }\limits_{{P_1} \in \left\{ {{P_{1C}},{ {\rca \max \left\{ {{\varphi _3},\omega } \right\} } },\eta } \right\}} {C_s}     & P_{1C} \in p_1
\\
\mathop {\arg }\limits_{{P_1}} \mathop {\max }\limits_{{P_1} \in \left\{ {  {\rca \max \left\{ {{\varphi _3},\omega } \right\} },\eta } \right\}} {C_s}     & P_{1C} \notin p_1
    \end{cases}
    \end{align}

		{\rca
		
		        \item If $S<0$, $T<0$, $U>0$ and {\rca $A''<0$, $B''<0$}
    \begin{align}
P_1^ \star = \begin{cases}
\mathop {\arg }\limits_{{P_1}} \mathop {\max }\limits_{{P_1} \in \left\{ {{P_{1C}},{\varphi _3},\eta' } \right\}} {C_s}     & P_{1C} \in p_1
\\
\mathop {\arg }\limits_{{P_1}} \mathop {\max }\limits_{{P_1} \in \left\{ {{\varphi _3},\eta' } \right\}} {C_s}     & P_{1C} \notin p_1
    \end{cases}
    \end{align}
}

    \item If $S>0$, $T>0$ and $U>0$
{\rca
    \begin{align}
P_1^\star = \begin{cases}
\eta   & A''>0, B''>0 
\\
\eta   & A''<0, B''>0
\\
\eta'  & A''<0, B''<0
\end{cases}
    \end{align}
}
  \end{enumerate}

\end{enumerate}
where ${P_{1C}} = \frac{{ - 2a\left( {1 + c} \right)\left( {1 + e} \right)U\sigma _n^2 - \sqrt \Delta  }}{{2aS}}$
and $\Delta =4abdQR{{\left( {\sigma _n^2} \right)^4}}\left( {1 + c} \right)\left( {1 + e} \right)$.
\label{thm 1 Lim Pow}
\end{thm}
\vspace{0.1cm}
\begin{proof}
The proof can be obtained in the similar way to that of Theorem~\ref{thm 1 ins}.
\end{proof}
{\rca For problem~\eqref{eqn:SC Opt Ins limited power eve dec}, the closed-form solution for $P_1$ is given in the following theorem.}
{\rca \begin{thm}
 Assume $f = {\left| {{h_{{U_1},{D_1}}}} \right|^2}$, $g = {\left| {{h_{{U_1},{D_2}}}} \right|^2}$,
$h = {\left| {{h_{{U_2},D_1}}} \right|^2}$, $i = {\left| {{h_{{U_2},D_2}}} \right|^2}$, $j = {\left| {{h_{{U_1},E}}} \right|^2}$, $E=ad + bc\gamma  - e(d + c\gamma )$, $F=ad - e(d + c\gamma )$, and $\tau  = \min \left( {{P_{{{\max }_2}}},\inf \rho} \right)$. Then, optimal $P_1$ is as follows:
\begin{enumerate}
  \item If $E<0$
    \begin{align}
&P_2^ \star  = 0
\end{align}
  \item If $ E > 0 $
  \begin{enumerate}
    \item If $F>0$ 
    \begin{align}
   \hspace{-1.5cm}  P_2^ \star  =
    \begin{cases}
P_{2C} & P_{2C} \in p_2
\\
\mathop {\arg }\limits_{{P_2}} \mathop {\max }\limits_{{P_2} \in \left\{ {0 ,\tau } \right\}} {C_s} & P_{2C} \notin p_2
  \end{cases}
  \end{align}
    \item If $F<0$
    \begin{align}
    P_2^ \star = \mathop {\arg }\limits_{{P_2}} \mathop {\max }\limits_{{P_2} 
		\in \left\{ { 0 ,\tau } \right\}} {C_s}
  \end{align}

\end{enumerate}
\end{enumerate}
where ${P_{2C}} = \frac{{ - G - \sqrt \Delta  }}{{2H}}$, $\Delta  = 4Eabcde\gamma (d + c\gamma ){\sigma ^2}$, 
$G= - 2bce\gamma (d + c\gamma )\sigma $, $H= - bce\gamma (ad + bc\gamma ) $, $\rho  = \frac{{\sigma _n^2\left( {{{\left| {{h_{{U_1},{D_1}}}} \right|}^2} - {{\left| {{h_{{U_1},E}}} \right|}^2}} \right){{\left| {{h_{{U_2},{D_2}}}} \right|}^2}}}{{\gamma {{\left| {{h_{{U_1},E}}} \right|}^2}{{\left| {{h_{{U_2},{D_1}}}} \right|}^2}{{\left| {{h_{{U_1},{D_2}}}} \right|}^2}}} - \frac{{\sigma _n^2}}{{{{\left| {{h_{{U_1},{D_2}}}} \right|}^2}}}$, and $p_2$ is the feasibility domain of the problem.
\label{thm ego eve dec}
\end{thm}}

\vspace{0.1cm}
{\rca \begin{proof}
The proof can be obtained in the similar way to that of Theorem~\ref{thm 1 ins}.
\end{proof}}
{\rc 
\section{Secrecy Energy Efficiency}  \label{sec:Secrecy energy efficiency}
Before going to Section~VI, we define a metric in order to investigate the energy efficiency of the considered scenario. We define the
secrecy energy efficiency, ${\eta _{SEE}}$, as the maximum secrecy rate obtained from the objective of (9), namely $\Psi$, to the optimal
consumed power of $U_1$, $P_1^{\star}$, ratio as ${\eta _{SEE}} = \frac{{\mathop {\max }\limits_{{P_1}} \Psi }}{{P_1^ \star }}$.
Similarly, in the case we have only one transceiver pair and an eavesdropper with no interfering 
user, the secrecy energy efficiency metric, ${\eta _{SEE}}$, can be defined as ${\eta _{SE{E_{su}}}} = \frac{{\log \left( {{{\left( {\sigma _n^2 + {P^ \star }{{\left| {{h_{U,D}}} \right|}^2}} \right)} \mathord{\left/
 {\vphantom {{\left( {\sigma _n^2 + {P^ \star }{{\left| {{h_{U,D}}} \right|}^2}} \right)} {\left( {\sigma _n^2 + {P^ \star }{{\left| {{h_{U,E}}} \right|}^2}} \right)}}} \right.
 \kern-\nulldelimiterspace} {\left( {\sigma _n^2 + {P^ \star }{{\left| {{h_{U,E}}} \right|}^2}} \right)}}} \right)}}{{{P^ \star }}}$,
  where $P$ is the transmission power, and $P^{\star}$ is the optimal transmission power obtained from the optimization problem in the nominator. When the condition ${\left| {{h_{U,D}}} \right|^2} > {\left| {{h_{U,E}}} \right|^2}$ holds, the optimum consumed power in the single-user case is $P_{max}$. In contrast, as we shall see in Section~VI, the optimal power consumed by $U_1$ in the interference channel is considerably lower than $P_{max}$. Hence, as shall be shown in Section~VI, in a wide range of powers, the interference network outperforms the single-user network in terms of secrecy energy efficiency.}
\section{Numerical Results} \label{sec:num results}
In this section, we present different scenarios as numerical examples to further clarify the derived
results. As a benchmark, we consider a single-user scenario where only one user is present in the environment and there is no second user to produce interference~\cite{Wyner:1975}. Then, we compare this benchmark with our system model. Here, we refer to the altruistic and egoistic scenarios as interference
channel modes. In all simulation scenarios, we assume that the noise power is equal to one,
i.e., $\sigma _n^2=1$. All the channel coefficients are modeled as i.i.d. complex normal random variables with real and imaginary parts being as $\mathcal{N}(0,1)$. The channel coefficients are normalized to have a unit variance as $\mathcal{CN}(0,1)$.

For the first scenario, we consider the effect of the users' power limits, $P_{{{\max }_1}}$ and
${P_{{{\max }_2}}}$, on the average secrecy rate as shown in Fig.~\ref{fig:ASR vs P1 and P2 Ins}
for $\text{SINR}=1$ at $U_2$'s destination. By observing the results in
Fig.~\ref{fig:ASR vs P1 and P2 Ins}, we can draw the following conclusions for both altruistic
and egoistic scenarios:
\begin{enumerate}
  \item Average secrecy rate of $U_1$ increases as $P_{{{\max }_1}}$ or $P_{{{\max }_2}}$ increases.

  \item Increasing $P_{{{\max }_1}}$ is more effective on improving the average secrecy rate rather
  than increasing $P_{{{\max }_2}}$. The reason is that increasing $U_2$'s power creates more interference
  to both $U_1$ and $E$.

  \item The average secrecy rate of $U_1$ is lower in the egoistic scenario since $U_2$ does not change its transmitted
  power in favor of $U_1$, and only adjusts it according to the required QoS at $D_2$. Also, by comparing Fig.~\ref{fig:ASR vs P1 and P2 Opt} and Fig.~\ref{fig:ASR vs P1 and P2 lim Opt}, it can be seen that $U_2$ consumes less power in the egoistic scenario. When the Cases~\eqref{subeq:P1} and~\eqref{subeq:P3} of
  Theorem~\ref{thm:Pos sec Ins} are true, $U_2$ can improve the secrecy rate by providing more power in the altruistic scenario.
\end{enumerate}

The average optimal powers consumed by $U_1$ and $U_2$ are shown in Fig.~\ref{fig:ASR vs P1 and P2 Opt} for the altruistic scenario. Following points
are implied by Fig.~\ref{fig:ASR vs P1 and P2 Opt} as:
\begin{enumerate}
  \item In contrast to the single-user case where the maximum power consumption is optimum when the data
  link is stronger than the wiretap link, average optimal powers expended
  by users in the interference channel modes are considerably less than the available quantity. So, the optimum power
  control in the interference channel leads to enormous power saving. \label{c1}

  \item As $P_{\max_1}$ increases, $U_2$ consumes more power. A higher power transmission from $U_1$ produces more interference
  on $D_2$. This makes $U_2$ to choose higher transmission power in order to maintain the QoS at $D_2$.

  \item As $P_{\max_2}$ increases, $U_1$ utilizes more power. A higher available power for $U_2$ enables it to compensate
  a higher interference from $U_1$, so $U_1$ transmits with a higher power to increase the secrecy rate. \label{c2}

  \item Depending on the maximum available power to the users, the optimal consumed power by one user
  can be higher or lower than the power consumed by the other user.
\end{enumerate}

Consumed powers by $U_1$ and $U_2$ for the egoistic scenario are illustrated in Fig.~\ref{fig:ASR vs P1 and P2 lim Opt}. As we can see,
the power consumption pattern is similar to the altruistic scenario as in Fig.~\ref{fig:ASR vs P1 and P2 Opt}. By
comparing Fig.~\ref{fig:ASR vs P1 and P2 lim Opt} with Fig.~\ref{fig:ASR vs P1 and P2 Opt}, it is noticed that the power consumed by the
users in the altruistic scenario is higher than the egoistic scenario.

Average excess SINR provided by $U_2$ at $D_2$ in the altruistic scenario is shown in Fig.~\ref{fig:ASR vs P1 and P2 QoS}
for different values of the required QoS, $\gamma_{th}$. Following messages are conveyed by
Fig.~\ref{fig:ASR vs P1 and P2 QoS} as:
\begin{enumerate}
  \item By increasing $P_{{{\max }_1}}$ for a fixed $P_{{{\max }_2}}$, the excess SINR provided at $D_2$ drops due
  to increased interference from $U_1$'s transmission.

  \item Increasing $P_{{{\max }_2}}$ for a fixed $P_{{{\max }_1}}$ leads to a higher excess SINR at $D_2$.
\end{enumerate}
\begin{figure}[t]
  \centering
  \includegraphics[width=8.5cm]{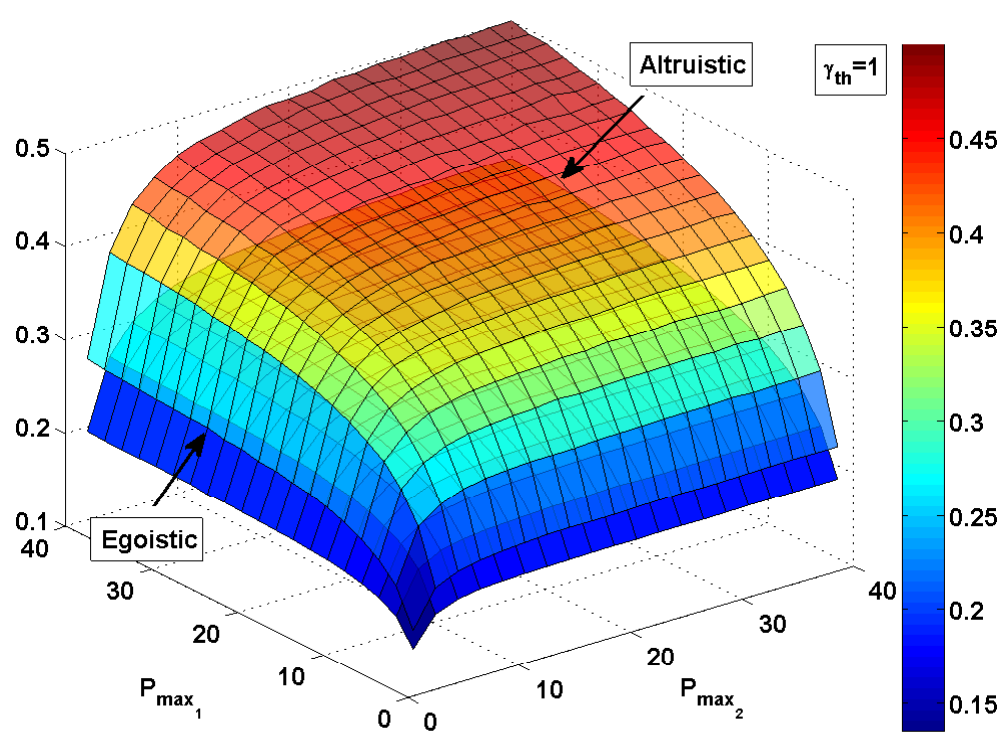}
  \caption{Average secrecy rate versus the users' maximum available powers in altruistic and egoistic scenarios.}
  \label{fig:ASR vs P1 and P2 Ins}
\end{figure}
\begin{figure}[t]
  \centering
  \includegraphics[width=8.5cm]{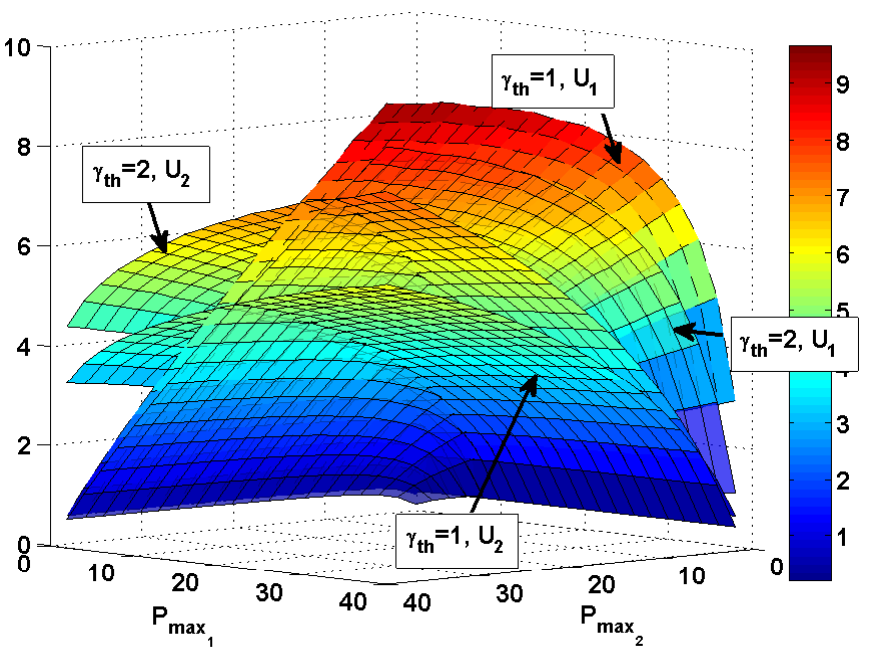}
  \caption{Average optimal power consumed by the users versus their maximum available powers in altruistic scenario.}
  \label{fig:ASR vs P1 and P2 Opt}
\end{figure}
\begin{figure}[t!]
  \centering
  \includegraphics[width=8.5cm]{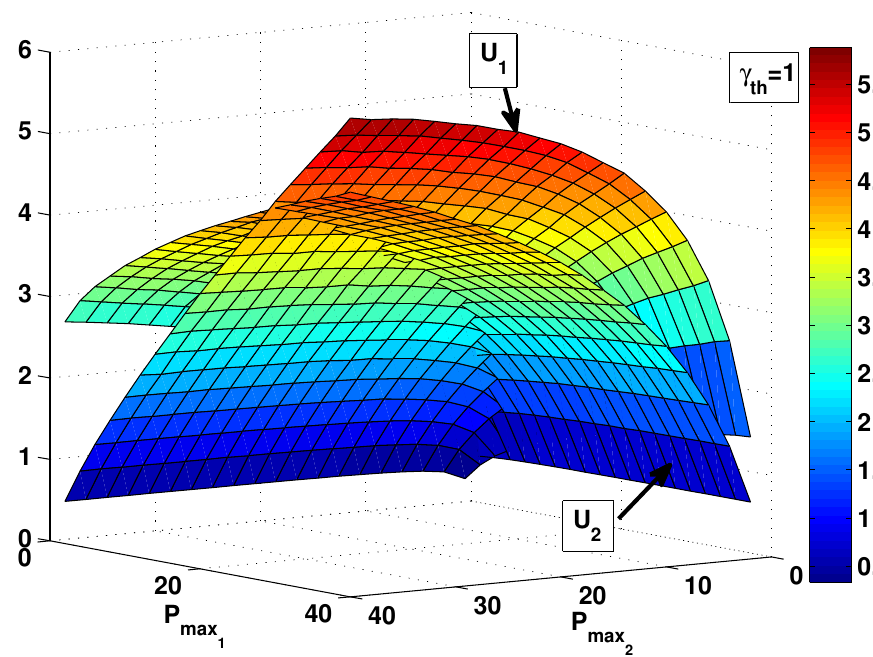}
  \caption{Average optimal power consumed by the users versus their maximum available powers in the egoistic scenario.}
  \label{fig:ASR vs P1 and P2 lim Opt}
\end{figure}
\begin{figure}[!ht]
  \centering
  \includegraphics[width=8.5cm]{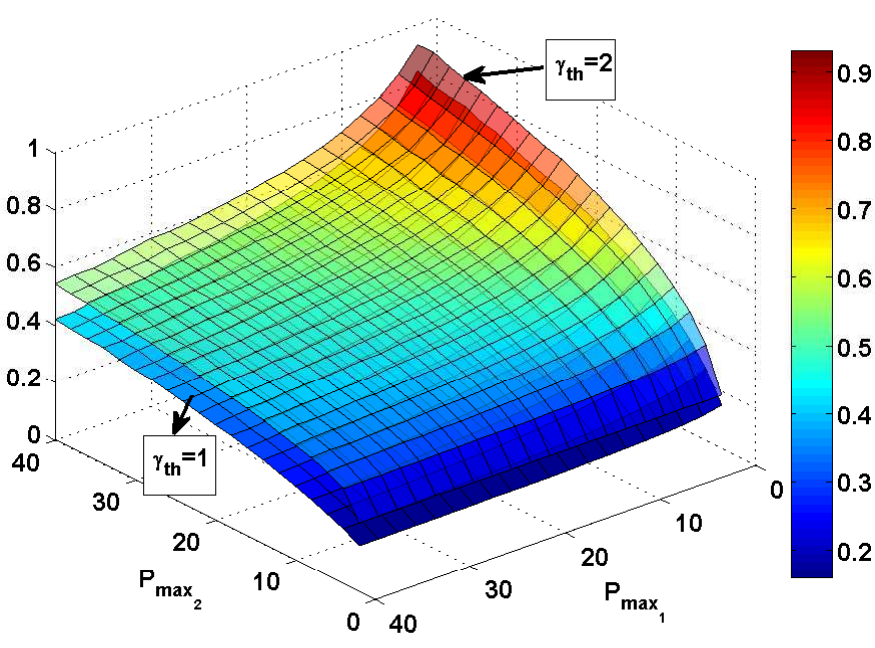}
  \caption{Average excess QoS provided at $D_2$ versus users' maximum available powers in the altruistic scenario.}
  \label{fig:ASR vs P1 and P2 QoS}
\end{figure}
\begin{figure}[t!]
  \centering
  \includegraphics[width=8.5cm]{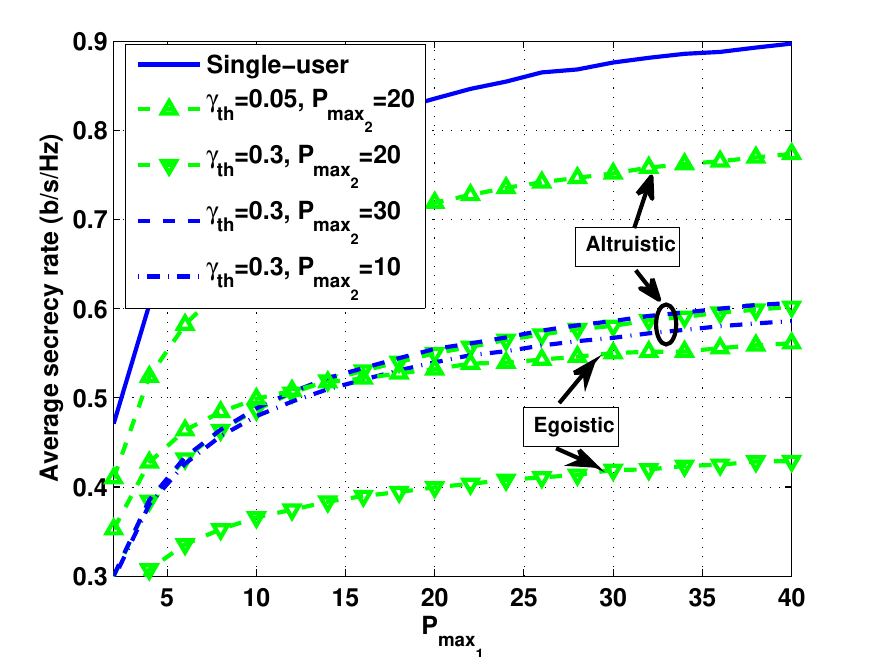}
  \caption{Average secrecy rate versus {\rc $U_1$}'s maximum available power.}
  \label{fig:ASR vs P1 Ins}
\end{figure}

The average secrecy rate comparison among the single-user benchmark and the interference channel modes is presented in
Fig.~\ref{fig:ASR vs P1 Ins} with respect to the maximum available power of {\rc $U_1$}. Following conclusions can be made according to
Fig.~\ref{fig:ASR vs P1 Ins}:
\begin{enumerate}

  \item Increasing $P_{\max_2}$ also enhances the average secrecy rate but much less compared to increasing the $P_{\max_1}$, because $U_2$ induces interference on both $D_1$ and $E$. \label{P2_2}

  \item The secrecy rate in the egoistic scenario is always lower than the one in the altruistic scenario. In the egoistic scenario, $U_2$ consumes power to only
  fulfil the QoS at $D_2$. As a result, $U_2$ does not increase its transmission power to produce interference on $E$ when the Cases~\eqref{subeq:P1} and~\eqref{subeq:P3}
  of Theorem~\ref{thm:Pos sec Ins} hold. However, in the altruistic scenario, $U_2$ can change its transmission power in favor of $U_1$ when it becomes necessary. \label{P2_3}

\end{enumerate}
A similar comparison as in Fig.~\ref{fig:ASR vs P1 Ins} is displayed in Fig.~\ref{fig:ASR vs P2 Ins} with respect
to the maximum available power of {\rc $U_2$}. {\rca The Statements~\ref{P2_2} and~\ref{P2_3}} of Fig.~\ref{fig:ASR vs P1 Ins} also 
hold for Fig.~\ref{fig:ASR vs P2 Ins}. As we see in Fig.~\ref{fig:ASR vs P2 Ins}, increasing $P_{\max_2}$ also 
increases the average secrecy rate. By increasing $P_{\max_2}$, $U_2$ gets a higher ability to suppress the 
interference coming from $U_1$ as well as causing more interference on $E$ when the Cases~\eqref{subeq:P1} 
and~\eqref{subeq:P3} of Theorem~\ref{thm:Pos sec Ins} hold. As a result, $U_1$ can transmit with a 
higher power and enhance the secrecy rate.
\begin{figure}[t]
  \centering
  \includegraphics[width=8.5cm]{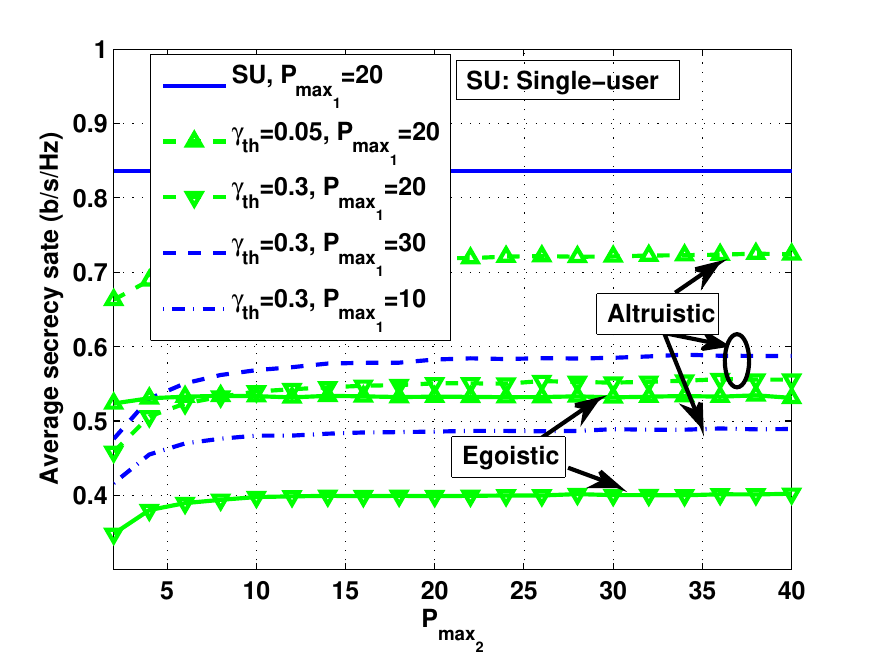}
  \caption{Average secrecy rate versus {\rc $U_2$}'s maximum available power.}
  \label{fig:ASR vs P2 Ins}
\end{figure}
\begin{figure}[t]
  \centering
  \includegraphics[width=8.5cm]{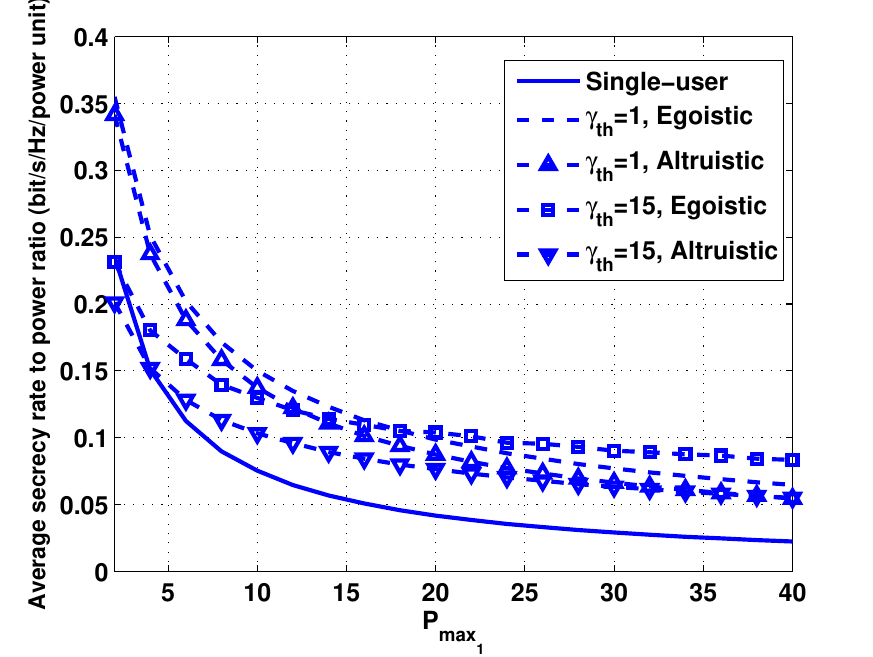}
  \caption{Average secrecy energy efficiency versus {\rc $U_1$}'s maximum available power.}
  \label{fig:ratio vs P1 Ins}
\end{figure}

As we can see from Fig.~\ref{fig:ASR vs P1 Ins} and Fig.~\ref{fig:ASR vs P2 Ins}, the average secrecy rate in
the interference channel modes is lower than its value in the single-user case. However, we should note that the power consumed in the interference channel modes
is considerably lower than the single-user case. To make a fair comparison, we use the ``secrecy energy
efficiency'' metric defined in Section~\ref{sec:Secrecy energy efficiency} to compare the secrecy rates of the interference 
channel modes and the single-user benchmark. This metric is derived for different values of $\gamma_{th}$ in Fig.~\ref{fig:ratio vs P1 Ins}. According
to graphs in Fig.~\ref{fig:ratio vs P1 Ins}, we can make the following conclusions:
\begin{enumerate}
  \item The secrecy energy efficiency is higher for the interference channel modes in a considerable range of $\gamma_{th}$ and $P_{\max_1}$. If we
  consider a specific available power for $U_1$, the acquired secrecy rate in the interference channel modes becomes higher than the one achieved in the single-user case.

  \item As the maximum available power to $U_1$ increases, the secrecy energy efficiency falls faster for the cases with lower $\gamma_{th}$.
\end{enumerate}
\section{Conclusion}  \label{sec:con}
We studied the effect of interference on improving the secrecy rate in a two-user wireless interference
network {\rc with input signals following Gaussian distribution}. We developed channel dependent expressions for both altruistic and egoistic scenarios to define the proper range
of transmission power for the interfering user, namely user $2$, in order to sustain a positive secrecy rate for the other user, namely user $1$. Closed-form solutions were obtained in order to perform joint
optimal power control for both users in the altruistic and egoistic scenarios.{\rca~It was shown that by decreasing the required QoS at user $2$'s destination, the secrecy rate in the interference channel improves and gets closer to the single-user case.}

Moreover, to fairly compare our scheme with the benchmark, the ratio of the secrecy rate over the optimal consumed power by user $1$ was introduced as a new metric called ``secrecy {\rc energy} efficiency'', in order to take into account both the secrecy rate and the consumed power. It was shown that in comparison with the single-user case, the secrecy {\rc energy} efficiency is considerably higher in the interference channel for a wide range of QoS at user $2$'s destination. 
\appendices
\section{Proof of Theorem~\ref{thm:Pos sec Ins}} \label{Appen Propo Ins}
For the objective function in~\eqref{eqn:SC Opt 2 Ins} to be positive, the following condition must hold
\begin{align}
&{\log _2}\left( {1 + \frac{{{P_1}{{\left| {{h_{{U_1},{D_1}}}} \right|}^2}}}{{{P_2}{{\left| {{h_{{U_2},{D_1}}}} \right|}^2} + \sigma _n^2}}} \right)
\nonumber\\
&- {\log _2}\left( {1 + \frac{{{P_1}{{\left| {{h_{{U_1},E}}} \right|}^2}}}{{{P_2}{{\left| {{h_{{U_2},E}}} \right|}^2} + \sigma _n^2}}} \right) > 0
\nonumber\\
&\Rightarrow  \frac{{{P_1}{{\left| {{h_{{U_1},{D_1}}}} \right|}^2}}}{{{P_2}{{\left| {{h_{{U_2},{D_1}}}} \right|}^2} + \sigma _n^2}} > \frac{{{P_1}{{\left| {{h_{{U_1},E}}} \right|}^2}}}{{{P_2}{{\left| {{h_{{U_2},E}}} \right|}^2} + \sigma _n^2}}
\nonumber\\
&\Rightarrow  \begin{cases}
{P_2} > \frac{{\sigma _n^2\left( {{{\left| {{h_{{U_1},E}}} \right|}^2} - {{\left| {{h_{{U_1},{D_1}}}} \right|}^2}} \right)}}
{B}\,\,\,\,\,\,\, & B > 0
\\
{P_2} < \frac{{\sigma _n^2\left( {{{\left| {{h_{{U_1},E}}} \right|}^2} - {{\left| {{h_{{U_1},{D_1}}}} \right|}^2}} \right)}}
{B}\,\,\,\,\,\,\, & B < 0
\end{cases}
\end{align}
where $B={{{{\left| {{h_{{U_1},{D_1}}}} \right|}^2}{{\left| {{h_{{U_2},E}}} \right|}^2} -
{{\left| {{h_{{U_2},{D_1}}}} \right|}^2}{{\left| {{h_{{U_1},E}}} \right|}^2}}}$.
\section{Proof of Theorem~\ref{thm 1 ins}} \label{Appen thm 1 ins}
In order to find the optimal $P_2$ for~\eqref{eqn:SC Opt 4 Ins}, we analyze the derivative of the objective
function in~\eqref{eqn:SC Opt 4 Ins}. The derivative is defined at the top of next page in~\eqref{eqn:OF Der 1} where
$a = P_{{\max }_1} {\left| {{h_{{U_1},{D_1}}}} \right|^2}$, $b = {\left| {{h_{{U_2},{D_1}}}} \right|^2}$, $c = P_{{\max }_1} {\left| {{h_{{U_1},E}}} \right|^2}$,
and $d = {\left| {{h_{{U_2},E}}} \right|^2}$. According to the sign of the derivative, the optimal $P_2$ can be found. The
denumerator in~\eqref{eqn:OF Der 1} is already positive, so the sign of~\eqref{eqn:OF Der 1} directly depends on the sign of the numerator. The numerator is a
quadratic equation. According to the sign of the discriminant of the quadratic equation~\cite[Section~5.1]{spiegel2008schaum}, denoted by
$\Delta  = 4abcd\sigma _n^2\left( {b - d} \right)\left[ { - d\left( {a + \sigma _n^2} \right) + b\left( {c + \sigma _n^2} \right)} \right]$, the status of the roots
can be defined. The sign of the discriminant can be defined as
\newcounter{MYtempeqncnt}
\begin{figure*}[]
\normalsize
\setcounter{MYtempeqncnt}{\value{equation}}
\setcounter{equation}{60}
\begin{align}
\frac{{\partial OF}}{{\partial {P_2}}} = \frac{{bd\left( {bc - ad} \right)P_2^2 + 2b\left( { - a + c} \right)d\sigma _n^2{P_2} +
 \sigma _n^2 \left( {cd {\sigma _n^2} - a\left( { - cd + b\left( {c + \sigma _n^2} \right)} \right)} \right)}}{{{{\left( {\sigma _n^2 + b{P_2}} \right)}^2}{{\left( {c + \sigma _n^2 + d{P_2}} \right)}^2}}}
\label{eqn:OF Der 1}
\end{align}
\setcounter{equation}{\value{MYtempeqncnt}}
\hrulefill
\vspace*{1pt}
\end{figure*}
\addtocounter{equation}{1}
\begin{figure*}[t]
        \centering

        \begin{subfigure}[b]{0.2\textwidth}
                \includegraphics[width=\textwidth]{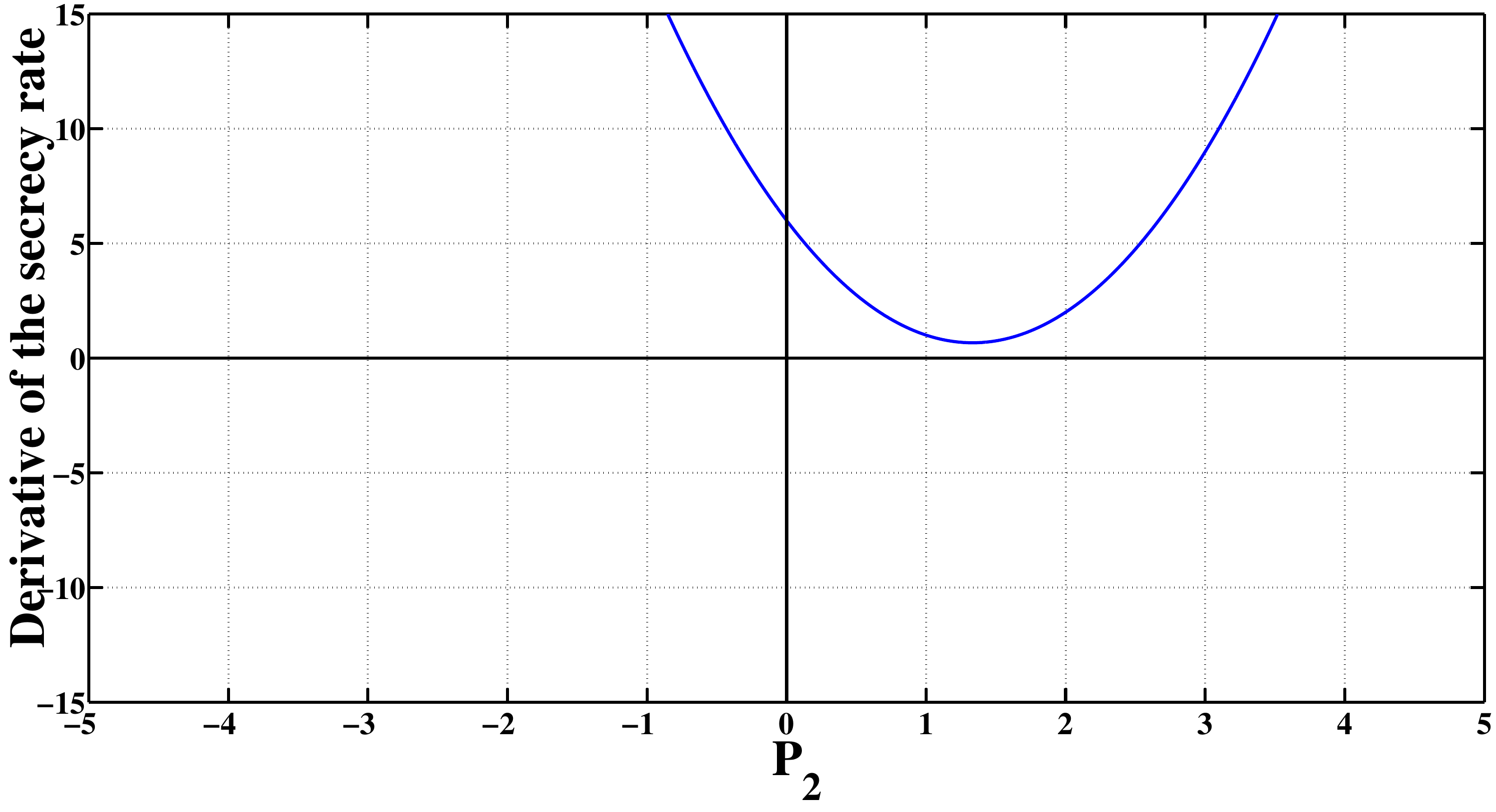}
                \caption{Strictly positive derivative.}
                \label{fig:St Pos}
        \end{subfigure}
        \begin{subfigure}[b]{0.2\textwidth}
                \includegraphics[width=\textwidth]{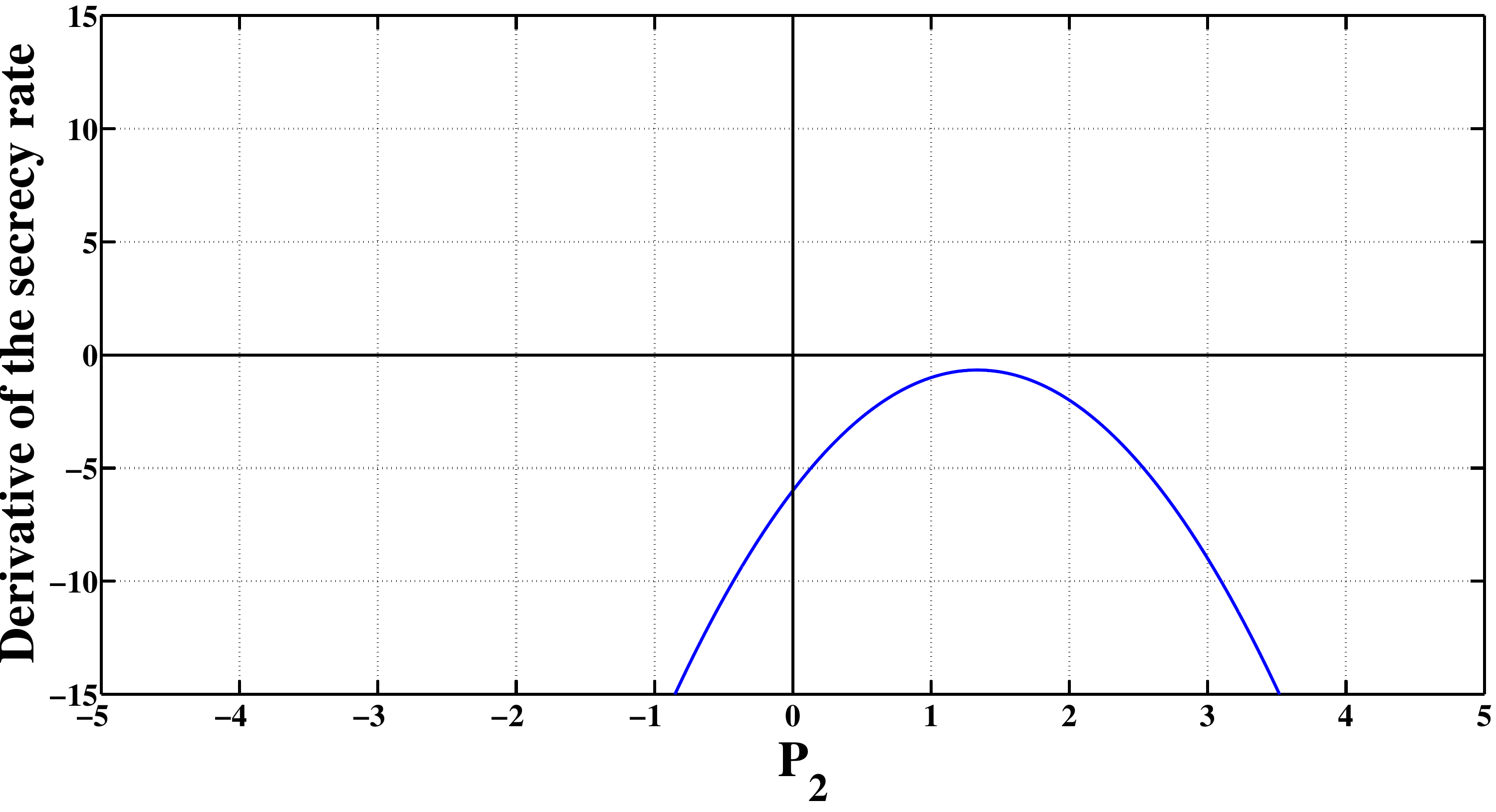}
                \caption{Strictly negative derivative.}
                \label{fig:St Neg}
        \end{subfigure}
        \begin{subfigure}[b]{0.2\textwidth}
                \includegraphics[width=\textwidth]{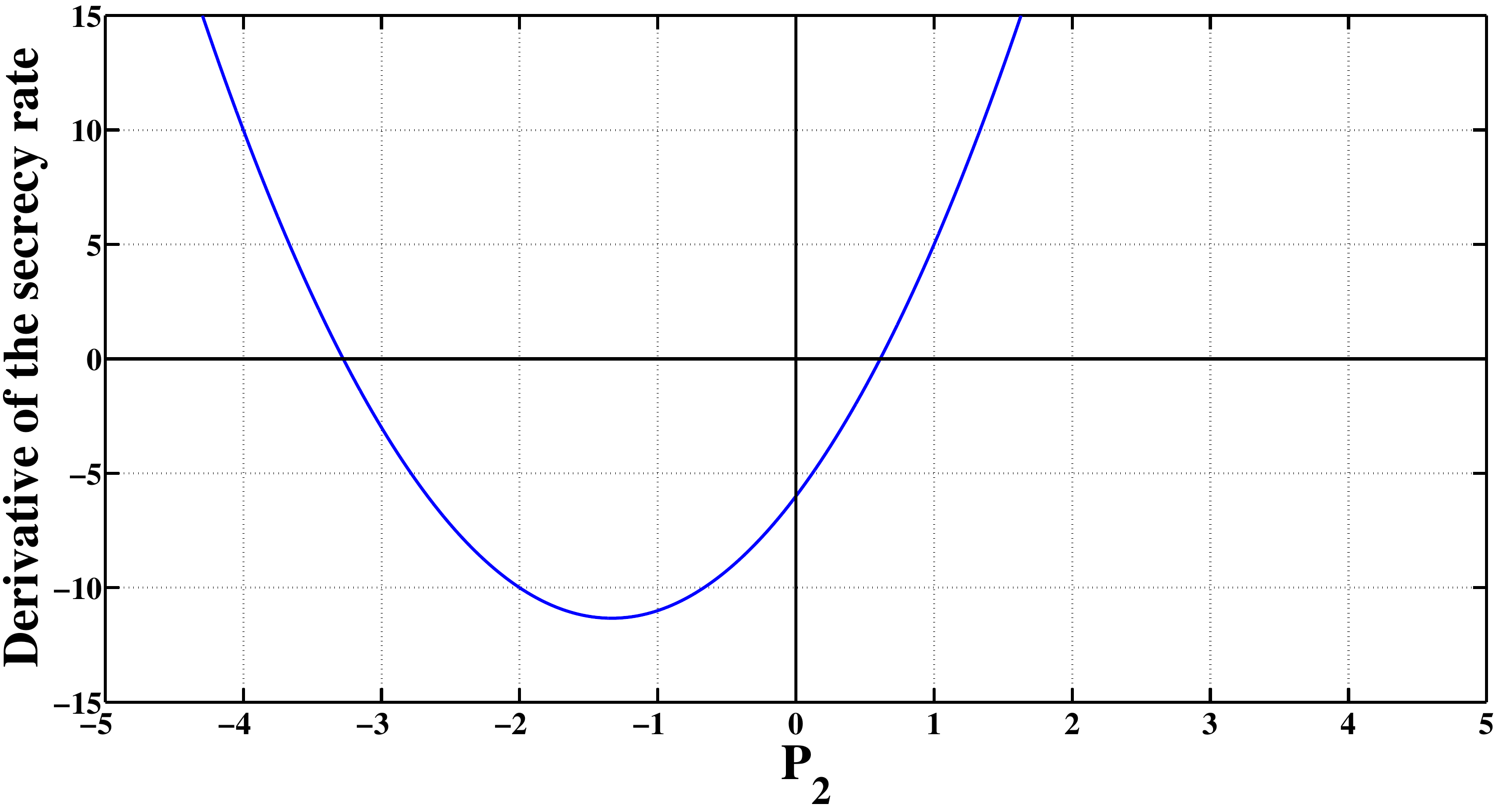}
                \caption{Negative and positive derivative.}
                \label{fig:One Pos root Neg Pos}
        \end{subfigure}
        \begin{subfigure}[b]{0.2\textwidth}
                \includegraphics[width=\textwidth]{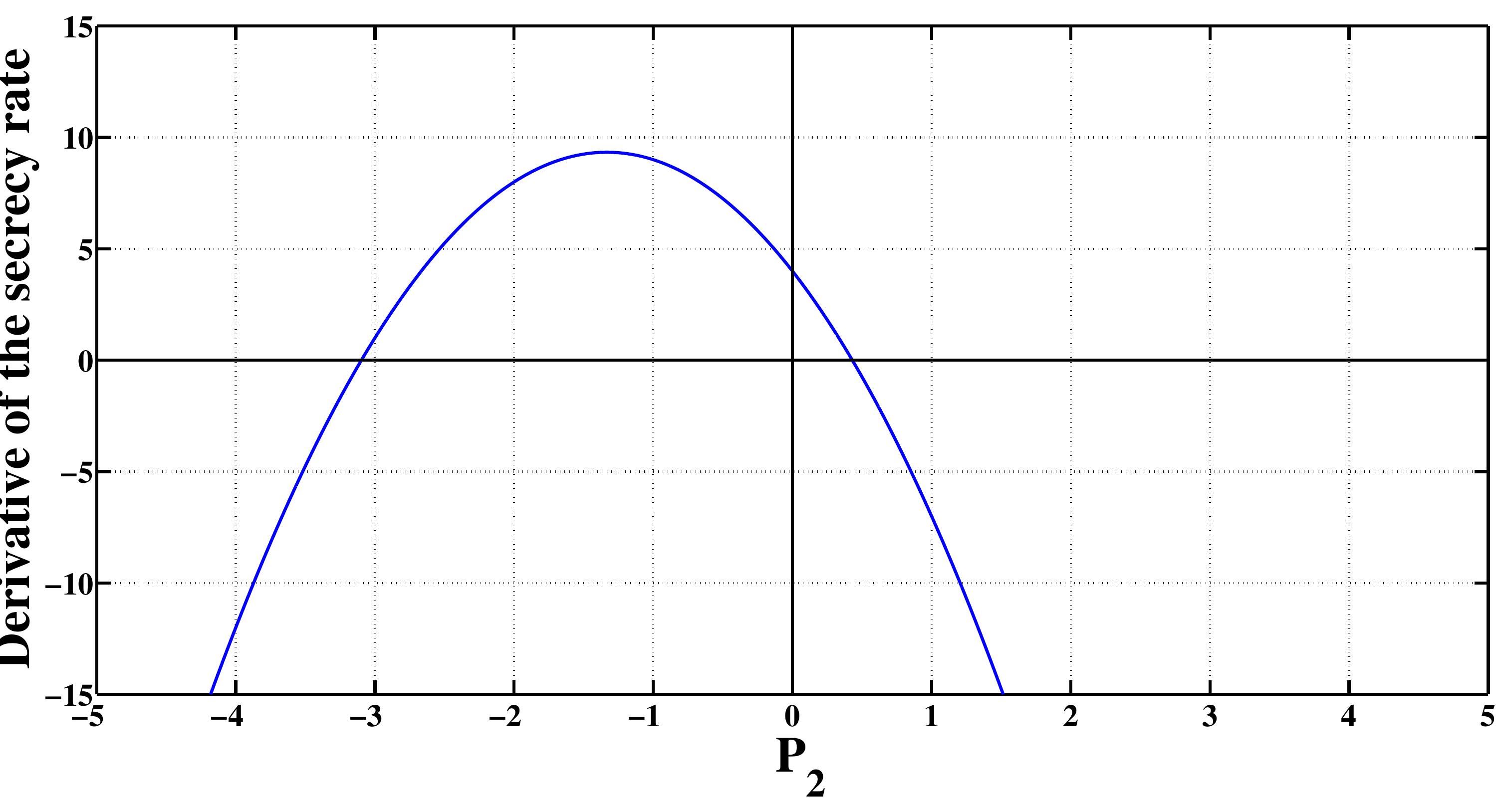}
                \caption{Positive and negative derivative.}
                \label{fig:One Pos root Pos Neg}
        \end{subfigure}
        \begin{subfigure}[b]{0.2\textwidth}
                \includegraphics[width=\textwidth]{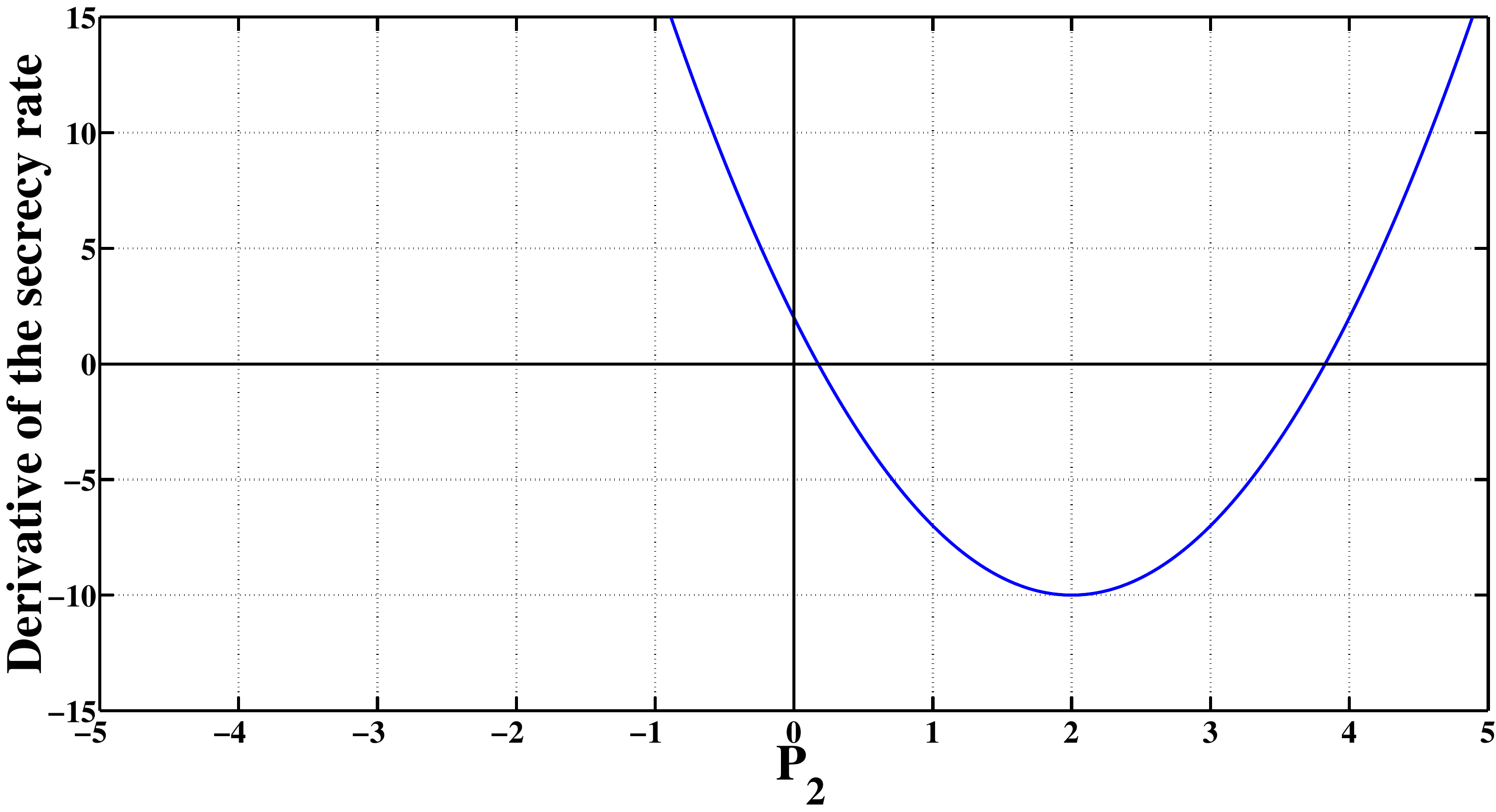}
                \caption{Positive, negative, positive derivative.}
                \label{fig:two Pos root Pos Neg Pos}
        \end{subfigure}
        \begin{subfigure}[b]{0.2\textwidth}
                \includegraphics[width=\textwidth]{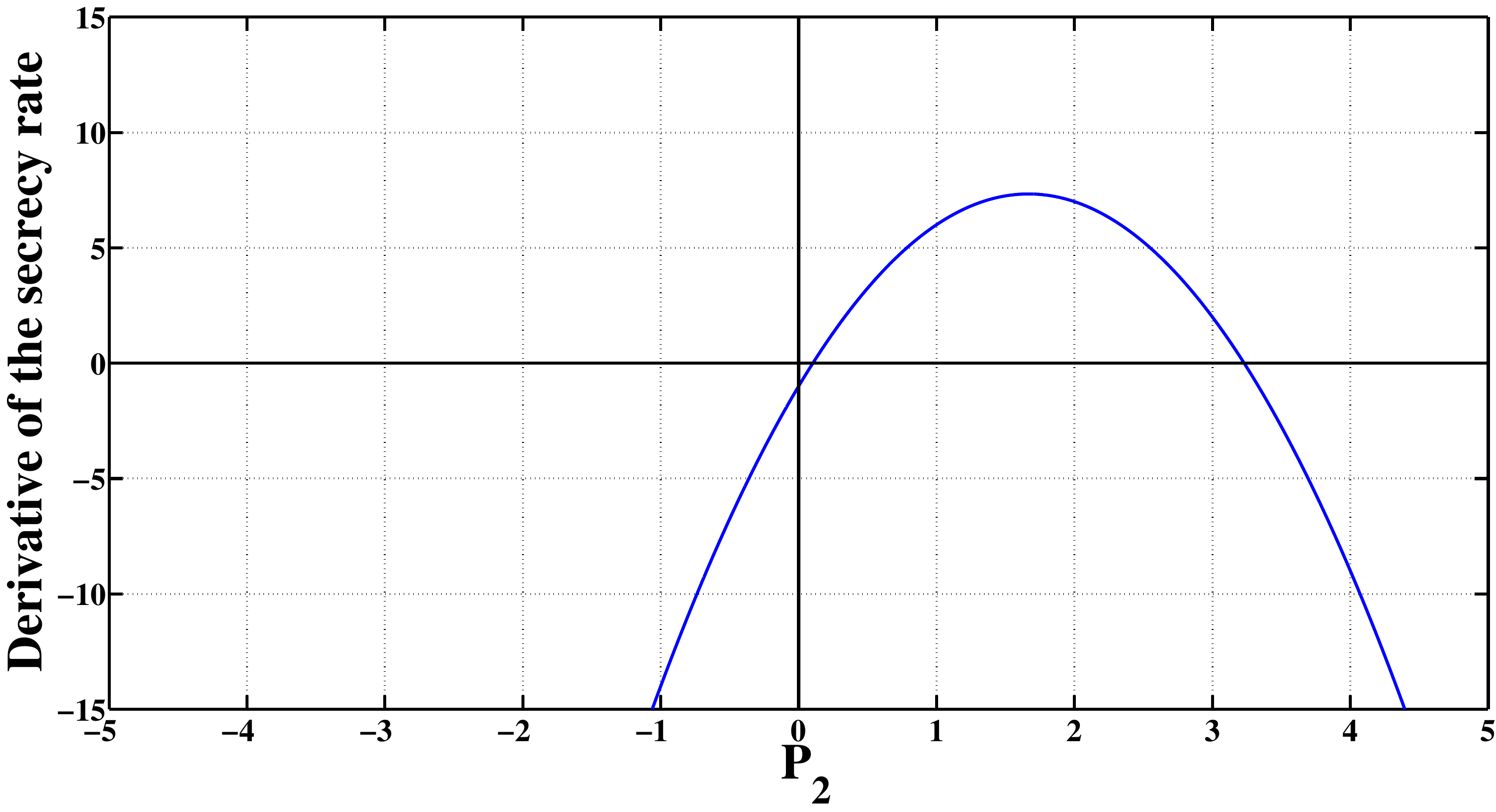}
                \caption{Negative, positive, negative derivative.}
                \label{fig:two Pos root Neg Pos Neg}
        \end{subfigure}
        \begin{subfigure}[b]{0.2\textwidth}
                \includegraphics[width=\textwidth]{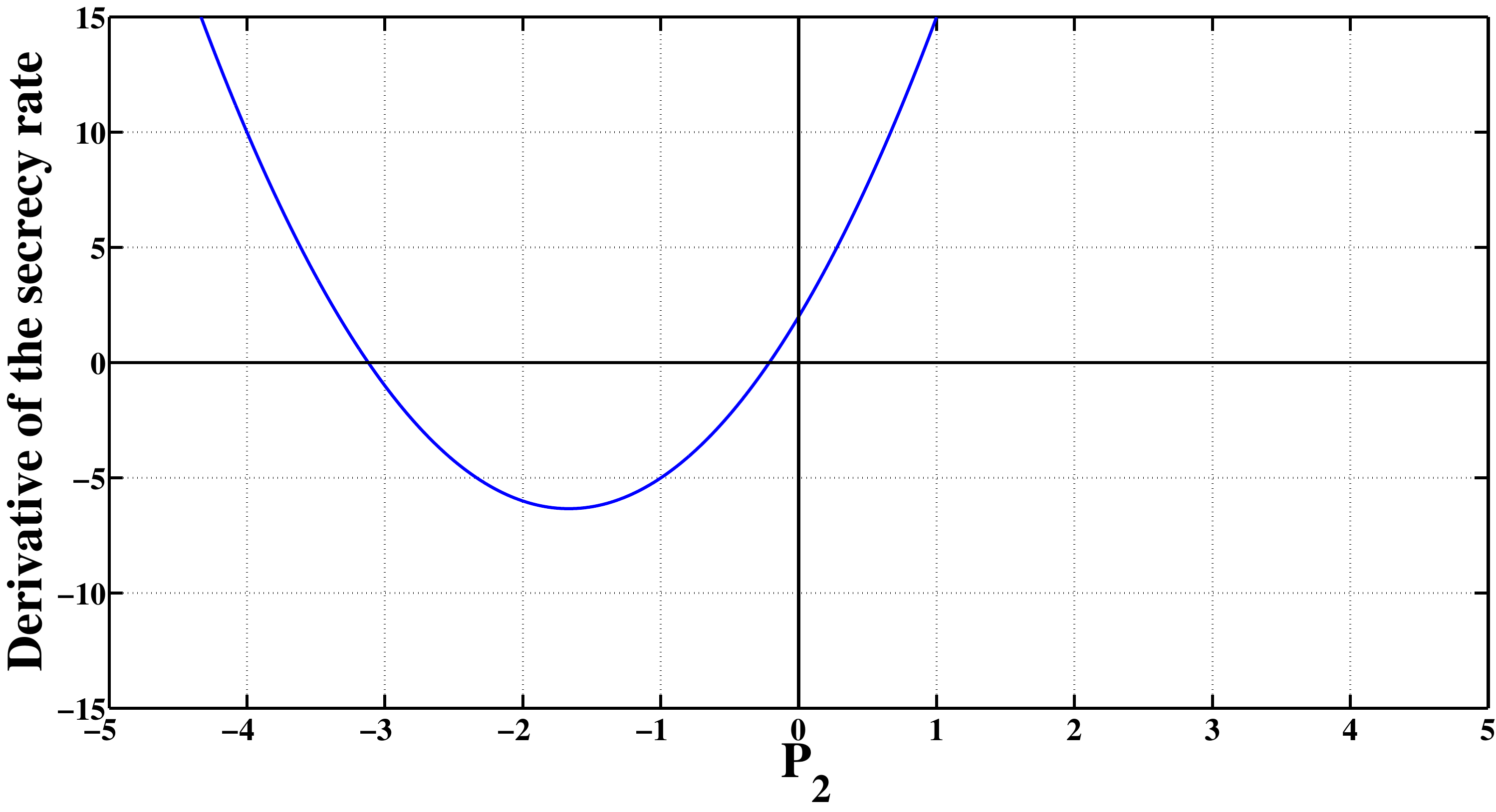}
                \caption{Negative, positive, negative derivative.}
                \label{fig:no Pos root Pos Neg Pos}
        \end{subfigure}
        \begin{subfigure}[b]{0.2\textwidth}
                \includegraphics[width=\textwidth]{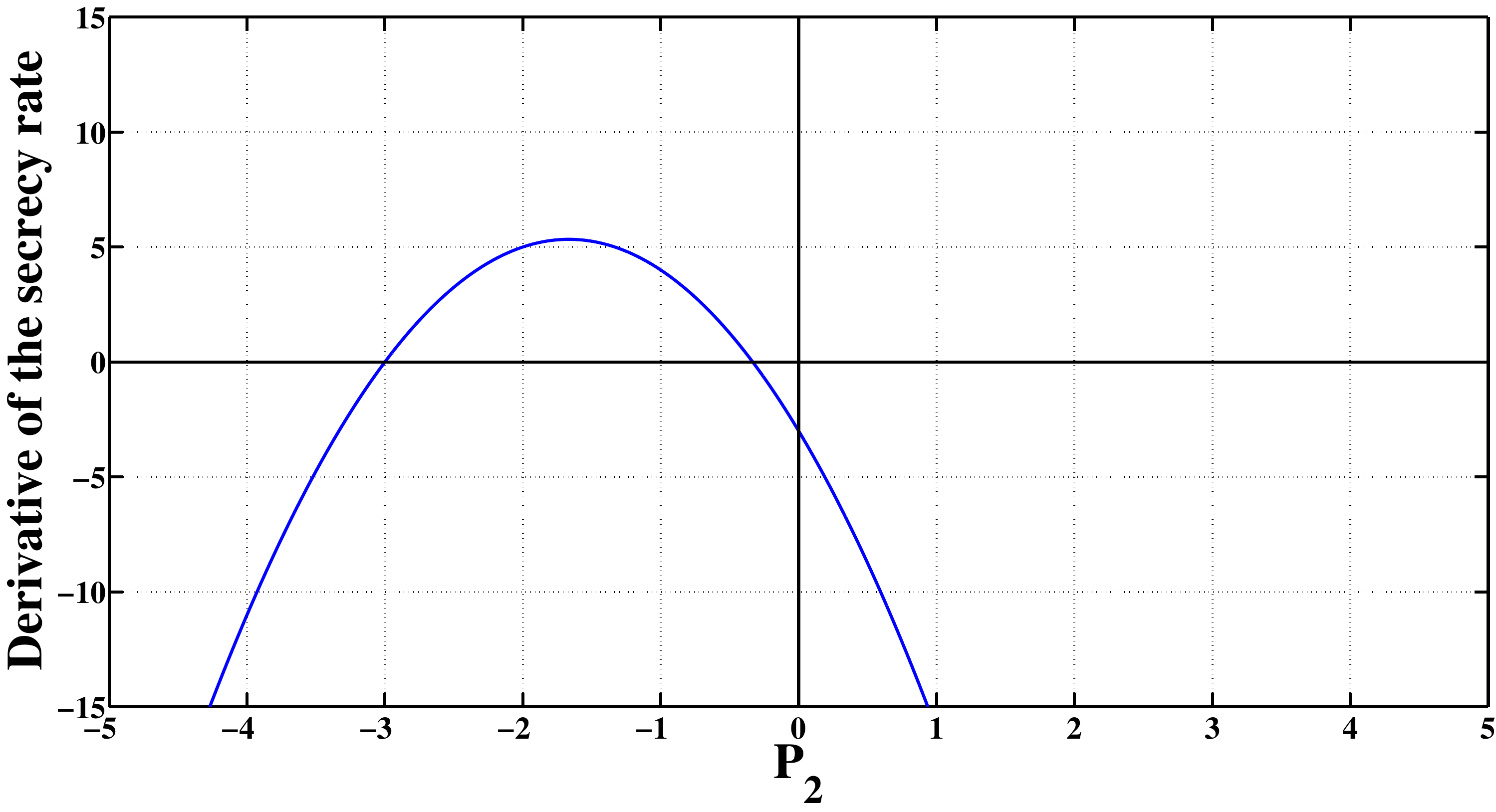}
                \caption{Negative, positive, negative derivative.}
                \label{fig:no Pos root Neg Pos Neg}
        \end{subfigure}
        \caption{Different cases for the sign of the derivative in~\eqref{eqn:OF Der 1}.}\label{fig:Der sign}
\end{figure*}
\begin{enumerate}
  \item If $\left( {b - d} \right)\left[ { - d\left( {a + \sigma _n^2} \right) + b\left( {c + \sigma _n^2} \right)} \right]<0$, $\Delta<0$.
  \item If $\left( {b - d} \right)\left[ { - d\left( {a + \sigma _n^2} \right) + b\left( {c + \sigma _n^2} \right)} \right]>0$, $\Delta>0$.
\end{enumerate}
Using the sign of $\Delta$ as well as the sign of the $P_2$'s coefficients in the quadratic equation which we denote them from
highest order to constant as $a'$, $b'$ and $c'$ in~\eqref{eqn:OF Der 1}, the sign of the derivative can be defined and
consequently the optimal value for $P_2$, $P_2^ \star$, can be found as follows:
\begin{enumerate}

  \item If $\Delta<0$, {\rc no root exists} for the numerator in~\eqref{eqn:OF Der 1} leading to the following cases:  \label{neg delta}

  \begin{enumerate}
    \item $a'>0$, then the derivative is strictly positive, as shown in Fig.~\ref{fig:St Pos}, and is monotonically increasing, so the highest value in the feasible set is the $P_2^ \star$.
    \item $a'<0$, then the derivative is strictly negative, as shown in Fig.~\ref{fig:St Neg}, and is monotonically decreasing, so the lowest possible value in the feasible set is the $P_2^ \star$.
  \end{enumerate}

  \item If $\Delta>0$, there exist two roots (two critical points for the objective function in~\eqref{eqn:SC Opt 4 Ins}) for the derivative leading to the following cases:

  \begin{enumerate}

  \item Only one of the roots is positive. This happens when the product of the roots~\cite[Section~5.1]{spiegel2008schaum}, $\frac{c'}{a'}$, is negative in the following cases: \label{one pos root}

       \begin{enumerate}

       \item $a'>0$ and $c'<0$, as shown in Fig.~\ref{fig:One Pos root Neg Pos}. In this case, the critical point is a minimum, so one of the vertices of the feasible
       domain is the $P_2^ \star$.

       \item $a'<0$ and $c'>0$, as shown in Fig.~\ref{fig:One Pos root Pos Neg}. For this case, the critical point is a maximum and if falls into the feasibility domain of $P_2$, it is the $P_2^ \star$. Otherwise, one of the vertices of
       the feasible domain is the $P_2^ \star$.

       \end{enumerate}

  \item Both of the roots are positive. This happens when both the product, $\frac{c'}{a'}$, and the sum~\cite[Section~5.1]{spiegel2008schaum}, $-\frac{b'}{a'}$, of
  the roots are positive in two following conditions: \label{both pos root}

      \begin{enumerate}

      \item $a'>0$, $c'>0$ and $b'<0$, as shown in Fig.~\ref{fig:two Pos root Pos Neg Pos}. For the first case, the derivative is first positive, then negative and then positive, respectively, meaning that the
      first root results in a maximum and the second root results in a minimum. If the smaller root falls in the feasibility domain of $P_2$, then by comparing it
      {\rc with} the vertices of the feasibility domain, {\rc $P_2^\star$ is found}. If the smaller root is not in the feasibility domain of $P_2$, the optimal value
      of $P_2$ is at one of the vertices of the feasibility domain.

      \item $a'<0$, $c'<0$ and $b'>0$, as shown in Fig.~\ref{fig:two Pos root Neg Pos Neg}. In this case, we find out that the larger root is a maximum. If the larger root falls in the feasibility
      domain of $P_2$, then by comparing it to the vertices of the feasibility domain of $P_2$, we can find the $P_2^\star$. If the larger root is not in the feasibility domain
      of $P_2$, we should find the optimal value of $P_2$ in the vertices of the feasibility domain.

      \end{enumerate}

  \item Both of the roots are negative. This happens when the product of the roots, $\frac{c'}{a'}$, is positive and the sum of the roots, $-\frac{b'}{a'}$, is negative in two
  following conditions: \label{two neg roots}

      \begin{enumerate}

      \item $a'>0$, $c'>0$ and $b'>0$, as shown in Fig.~\ref{fig:no Pos root Pos Neg Pos}. Since the transmission power is always positive, the critical points cannot be the answer to $P_2^\star$. For the first case, the derivative is first positive, then negative and then positive, respectively. As a result, the secrecy rate will be increasing after $P_2>0$. So, $P_2^ \star$ is the maximum possible value of $P_2$ inside the feasibility set. \label{11}

      \item $a'<0$, $c'<0$ and $b'<0$, as shown in Fig.~\ref{fig:no Pos root Neg Pos Neg}. As in Case~\ref{11}, the critical points cannot be the answer to $P_2^{\star}$. For the first case, the derivative is first negative, then positive and then negative, respectively. So, the secrecy rate is decreasing after $P_2>0$. Hence, $P_2^{\star}$ is the minimum possible value of $P_2$ inside the feasibility set.

      \end{enumerate}

  \end{enumerate}

\end{enumerate}
{\rc   In deriving the above closed-form optimal solutions, we have considered all the possible cases of discriminant sign, $\Delta$, and the coefficients of the quadratic equation, $a'$, $b'$, and $c'$. In each case, we have calculated all the critical points and if applicable, these critical points are compared with the vertices of the domain to make sure that the derived power value is globally optimum. Hence, the optimal solutions presented in Appendix~\ref{Appen thm 1 ins} are global optimum. }

\begin{IEEEbiography}
    [{\includegraphics[width=1in,height=1.25in,clip,keepaspectratio]{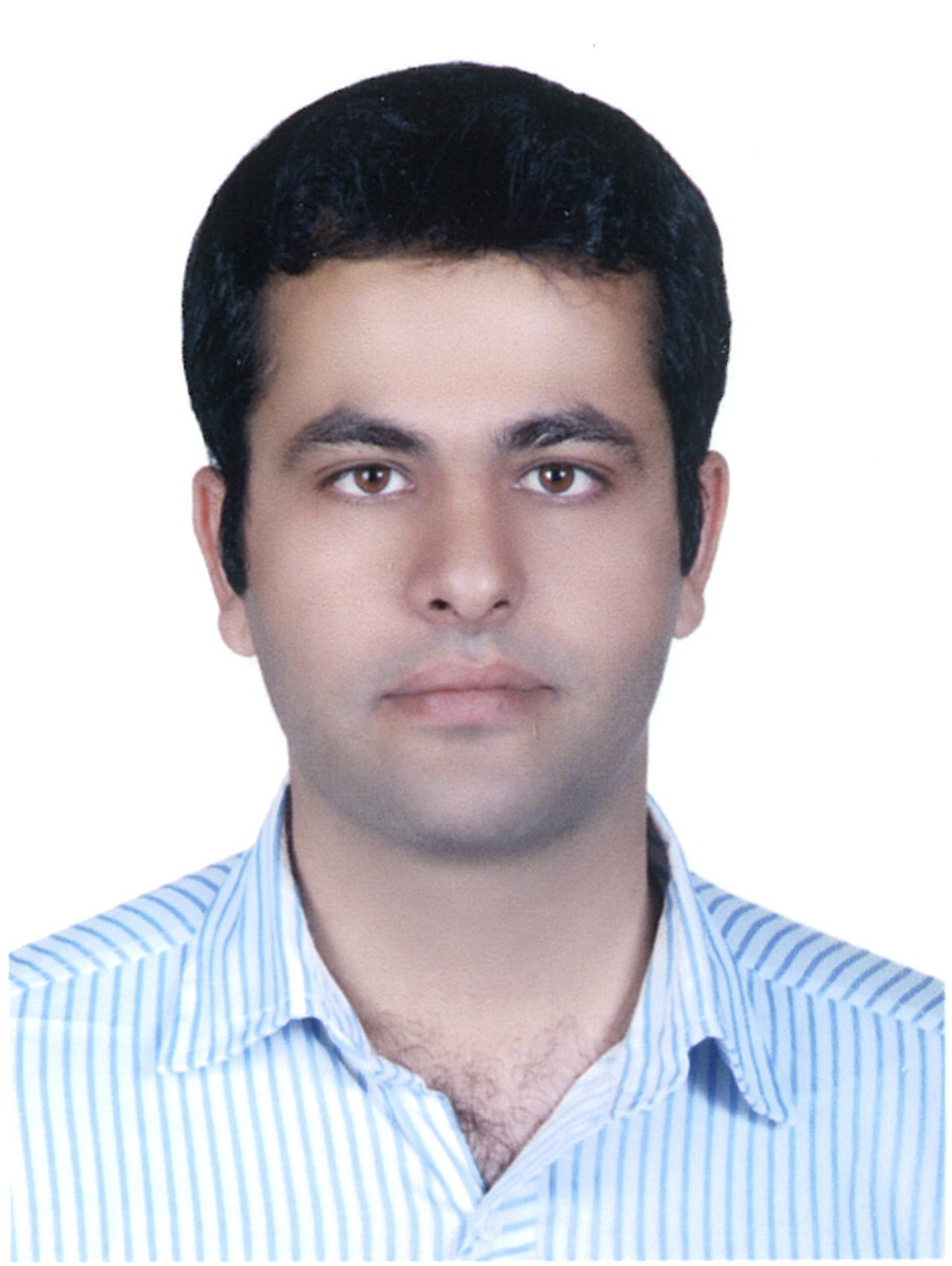}}]{Ashkan Kalantari}
Ashkan Kalantari (AK) was born in Iran.  He received his BSc and MSc degrees from K. N. Toosi 
University of Technology, Tehran, Iran in 2009 and 2012, respectively. He is currently working toward the 
Ph.D. degree with the research group SIGCOM in the Interdisciplinary Centre for Security, Reliability and 
Trust (SnT), University of Luxembourg. His research interest is physical layer security in wireless and satellite communications.
\end{IEEEbiography}
\begin{IEEEbiography}
    [{\includegraphics[width=1in,height=1.25in,clip,keepaspectratio]{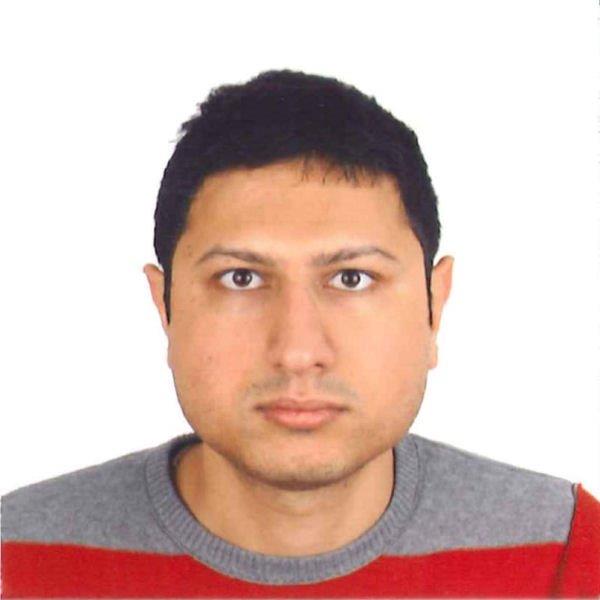}}]{Sina Maleki} 
		received his BSc degree from Iran University of Science and Technology, Tehran, Iran in 2006, and MSc and PhD degrees from Delft University of Technology, Delft, The Netherlands, in 2009 and 2013, respectively. From July 2008 to April 2009, he was an intern student at the Philips Research Center, Eindhoven, The Netherlands, working on spectrum sensing for cognitive radio networks. Since August 2013, he has been working at the Interdisciplinary Centre for Security, Reliability and Trust, University of Luxembourg, where he is working on cognitive radio for satellite communications within the EU FP7 CoRaSat project, and EU H2020 SANSA, as well as Luxembourgish national projects SATSENT, and SeMIGod.
\end{IEEEbiography}
\begin{IEEEbiography}
    [{\includegraphics[width=1in,height=1.25in,clip,keepaspectratio]{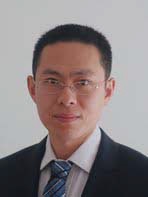}}]{Gan Zheng} 
		(S'05-M'09-SM'12) is currently a Lecturer in School of Computer Science and Electronic Engineering, University of Essex, UK. 
He received the B. E. and the M. E. from Tianjin University, Tianjin, China, in 2002 and 2004,
respectively, both in Electronic and Information Engineering,and the PhD 
degree in Electrical and Electronic Engineering from The University 
of Hong Kong, Hong Kong, in 2008. Before he joined University of Essex, he worked as a Research 
Associate at University College London, UK, and University of Luxembourg,
Luxembourg. His research interests include cooperative communications, cognitive radio, physical-layer security, full-duplex radio and energy harvesting.  He is the 
first recipient for the 2013 IEEE Signal Processing Letters Best Paper Award.
\end{IEEEbiography}
\begin{IEEEbiography}
    [{\includegraphics[width=1in,height=1.25in,clip,keepaspectratio]{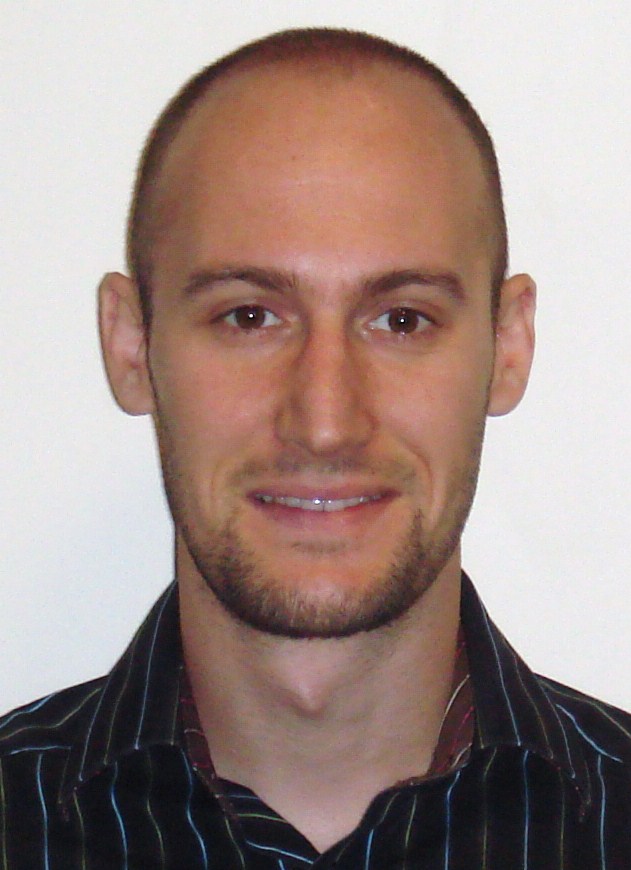}}]{Symeon Chatzinotas}
 (MEng, MSc, PhD, SMIEEE) received the M.Eng. in Telecommunications from Aristotle University of Thessaloniki, Greece and the M.Sc. and Ph.D. in Electronic Engineering from University of Surrey, UK in 2003, 2006 and 2009 respectively. He is currently a Research Scientist with the research group SIGCOM in the Interdisciplinary Centre for Security, Reliability and Trust, University of Luxembourg, managing H2020, ESA and FNR projects. In the past, he has worked in numerous R\&D projects for the Institute of Informatics \& Telecommunications, National Center for Scientific Research “Demokritos,” the Institute of Telematics and Informatics, Center of Research and Technology Hellas and Mobile Communications Research Group, Center of Communication Systems Research, University of Surrey. He has authored more than 120 technical papers in refereed international journals, conferences and scientific books. His research interests are on multiuser information theory, cooperative/cognitive communications and wireless networks optimization. Dr Chatzinotas is the co-recipient of the 2014 Distinguished Contributions to Satellite Communications Award, Satellite and Space Communications Technical Committee, IEEE Communications Society. He is currently co-editing a book on "Cooperative and Cognitive Satellite Systems" to appear in 2015 by Elsevier and he is co-organizing the First International Workshop on “Cognitive Radios and Networks for Spectrum Coexistence of Satellite and Terrestrial Systems” (CogRaN-Sat) in conjunction with the IEEE ICC 2015, 8-12 June 2015, London, UK.
\end{IEEEbiography}
\begin{IEEEbiography}
    [{\includegraphics[width=1in,height=1.25in,clip,keepaspectratio]{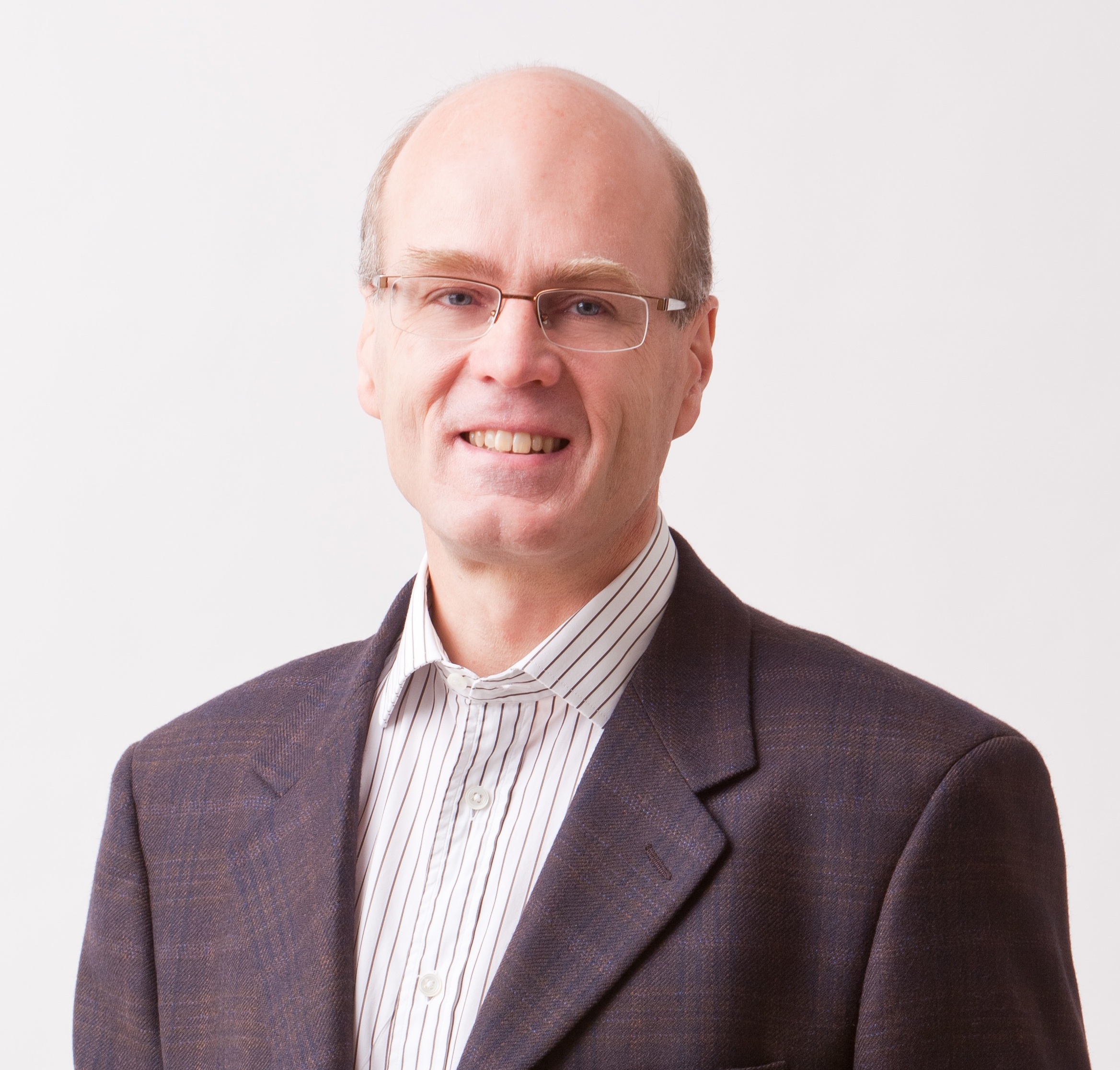}}]{Bj\"{o}rn Ottersten}
was born in Stockholm, Sweden, 1961. He received the M.S. degree in electrical engineering and applied physics from Link\"{o}ping University, Link\"{o}ping, Sweden, in 1986. In 1989 he received the Ph.D. degree in electrical engineering from Stanford University, Stanford, CA. Dr. Ottersten has held research positions at the Department of Electrical Engineering, Link\"{o}ping University, the Information Systems Laboratory, Stanford University, the Katholieke Universiteit Leuven, Leuven, and the University of Luxembourg. During 96/97 Dr. Ottersten was Director of Research at ArrayComm Inc, a start-up in San Jose, California based on Ottersten’s patented technology. He has co-authored journal papers that received the IEEE Signal Processing Society Best Paper Award in 1993, 2001, 2006, and 2013 and 3 IEEE conference papers receiving Best Paper Awards. In 1991 he was appointed Professor of Signal Processing at the Royal Institute of Technology (KTH), Stockholm. From 1992 to 2004 he was head of the department for Signals, Sensors, and Systems at KTH and from 2004 to 2008 he was dean of the School of Electrical Engineering at KTH. Currently, Dr. Ottersten is Director for the Interdisciplinary Centre for Security, Reliability and Trust at the University of Luxembourg. Dr. Ottersten is a board member of the Swedish Research Council and as Digital Champion of Luxembourg, he acts as an adviser to the European Commission. Dr. Ottersten has served as Associate Editor for the IEEE Transactions on Signal Processing and on the editorial board of IEEE Signal Processing Magazine. He is currently editor in chief of EURASIP Signal Processing Journal and a member of the editorial boards of EURASIP Journal of Applied Signal Processing and Foundations and Trends in Signal Processing. Dr. Ottersten is a Fellow of the IEEE and EURASIP and a member of the IEEE Signal Processing Society Board of Governors. In 2011 he received the IEEE Signal Processing Society Technical Achievement Award. He is a first recipient of the European Research Council advanced research grant. His research interests include security and trust, reliable wireless communications, and statistical signal processing.
\end{IEEEbiography}
\end{document}